\documentclass[runningheads]{llncs}

\usepackage[utf8]{inputenc}
\usepackage{complexity}
\usepackage{amsmath, amssymb}
\usepackage{macros}
\usepackage{tikz,pgfplots}    
\usetikzlibrary{automata,petri,positioning,shapes,calc,
	decorations.pathmorphing,patterns,arrows,fit,backgrounds,
	shapes.multipart,fadings,shapes.geometric}
\usepackage{thm-restate}
\usepackage{bm}

\usepackage{xcolor}
\definecolor{lightergray}{gray}{0.8}
\definecolor{darkgreen}{rgb}{0.0, 0.6, 0.0}

\usepackage[hidelinks, colorlinks, citecolor=darkgreen]{hyperref}
\usepackage{cleveref}
\usepackage{caption}
\usepackage{subcaption}
\usepackage{colortbl}
\usepackage{enumitem}
\usepackage{stmaryrd}
\usepackage{mathtools}\usepackage[ruled, noend]{algorithm2e} 
\SetKwRepeat{Do}{do}{while}            
\newcommand{\algskip}{\vspace{3pt}}    
\usepackage{thm-restate}
\usepackage{orcidlink}                 
\renewcommand\orcidID[1]{\orcidlink{#1}}

\usepackage[firstpage]{draftwatermark}
\SetWatermarkText{\hspace*{4in}\raisebox{4in}{\includegraphics[scale=0.1]{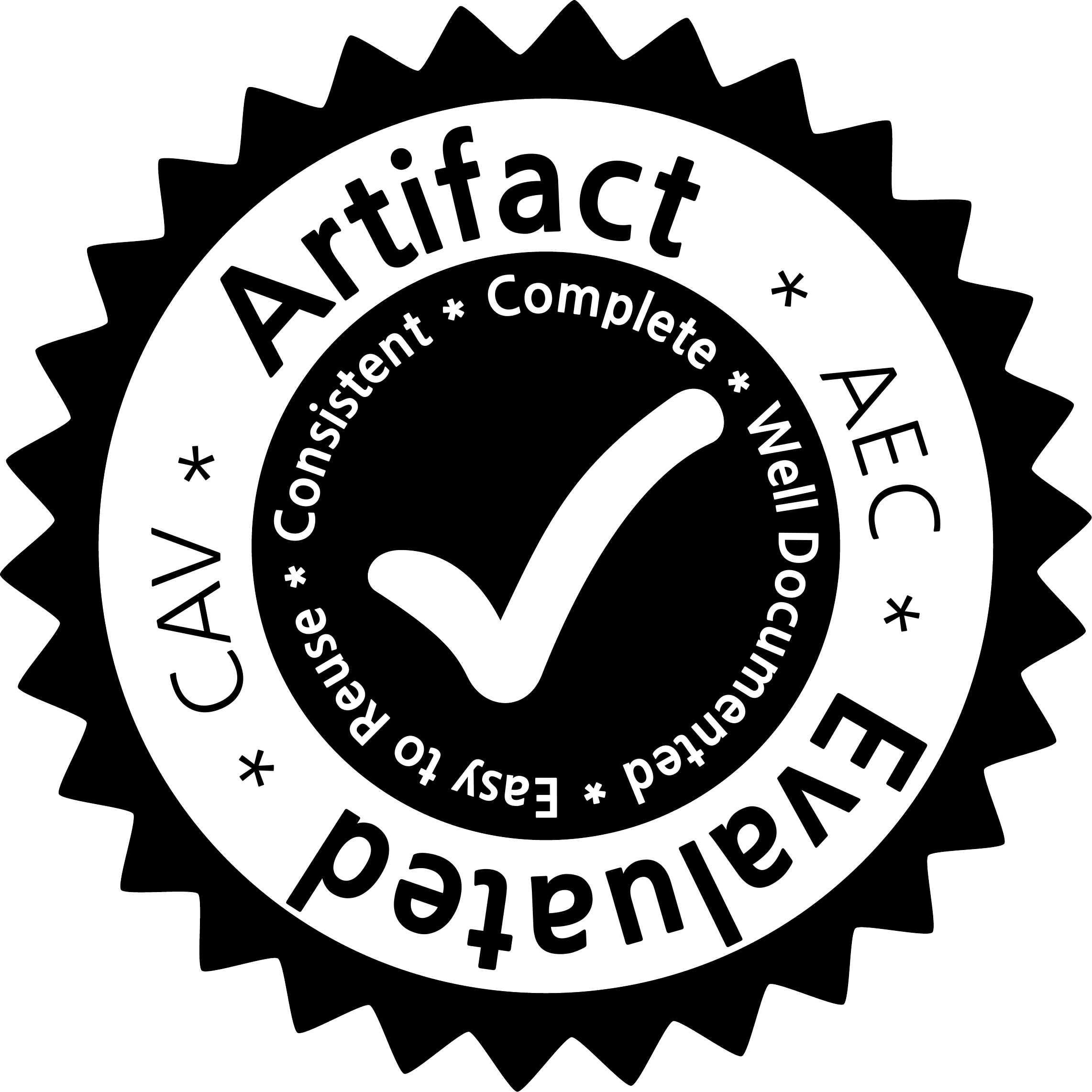}}}
\SetWatermarkAngle{0}

\DeclareMathAlphabet{\pazocal}{OMS}{zplm}{m}{n}

\newcommand{\TOWER}{\textsf{TOWER}}

\renewcommand{\C}{\mathcal{C}}
\renewcommand{\S}{\mathcal{S}}
\newcommand{\calP}{\pazocal{P}}

\newcommand{\tr}{\mapsto}

\newcommand{\sem}[1]{\llbracket#1\rrbracket}
\newcommand{\disa}[1]{\text{dis}(#1)}
\renewcommand{\sup}[1]{\text{sup}(#1)}
\renewcommand{\inf}[1]{\text{inf}(#1)}
\newcommand{\dead}[1]{\text{dead}(#1)}

\newcommand{\enab}{\textit{enab}}

\newcommand{\Leaves}{\pazocal{L}}
\newcommand{\Bottom}{\pazocal{B}}
\newcommand{\Epres}{Presburger}
\def\Next{{\mathbf{X}}}
\def\Until{\mathbin{\mathbf{U}}}
\DeclareMathOperator{\Down}{
	\mathchoice{\downarrow}
	{\downarrow}
	{\downarrow \!}
	{\downarrow \!}
}
\DeclareMathOperator{\Up}{
	\mathchoice{\uparrow}
	{\uparrow}
	{\uparrow \!}
	{\uparrow \!}
}

\newcommand{\true}{\textit{true}}

\newcommand{\Dead}{\textit{Dead}}

\newcommand{\Stage}{\pazocal{S}}

\newcommand{\ms}[1]{\langle #1 \rangle}

\newcommand{\AY}{\text{A}_\text{Y}}
\newcommand{\AN}{\text{A}_\text{N}}
\newcommand{\PY}{\text{P}_\text{Y}}
\newcommand{\PN}{\text{P}_\text{N}}
\newcommand{\DeathCert}{\mathit{DeathCert}}
\newcommand{\Init}{{\mathcal{I}}}

\begin{document}
\title{Checking Qualitative Liveness Properties of Replicated Systems with Stochastic Scheduling\thanks{
Michael Blondin is supported by a Discovery Grant from the Natural
Sciences and Engineering Research Council of Canada (NSERC) and by the
Fonds de recherche du Qu\'{e}bec -- Nature et technologies (FRQNT).
Javier Esparza, Martin Helfrich and Philipp J. Meyer have received funding from the European Research Council (ERC) under the European Union's Horizon 2020 research and innovation programme under grant agreement No~787367 (PaVeS).
Anton\'{\i}n Ku\v{c}era is supported by the Czech Science Foundation, grant No.~18-11193S.}}
\titlerunning{Checking Qualitative Liveness Properties of Stochastic Replicated Systems}
%
\author{Michael Blondin\inst{1}\orcidID{0000-0003-2914-2734} \and
Javier Esparza\inst{2}\orcidID{0000-0001-9862-4919} \and
Martin Helfrich\inst{2}\orcidID{0000-0002-3191-8098} \and \\
Anton\'{\i}n Ku\v{c}era\inst{3}\orcidID{0000-0002-6602-8028} \and
Philipp J. Meyer\inst{2}\orcidID{0000-0003-1334-9079}}
\authorrunning{M. Blondin et al.}
\institute{
Universit\'{e} de Sherbrooke, Sherbrooke, Canada \\
\email{michael.blondin@usherbrooke.ca} \and
Technical University of Munich, Munich, Germany \\
\email{\{esparza,helfrich,meyerphi\}@in.tum.de}
\and
Masaryk University, Brno, Czechia \\
\email{tony@fi.muni.cz}
}
%
\maketitle              
\begin{abstract}
We present a sound and complete method for the verification of 
qualitative liveness properties of replicated systems under stochastic scheduling. 
These are systems consisting of a finite-state program, executed by an unknown number of indistinguishable agents, where the next agent to make a move is determined by the result of a 
random experiment. We show that if a property of such a system holds, then there is always a witness in the shape of a \emph{Presburger stage graph}: a finite graph whose nodes are Presburger-definable sets of configurations.
Due to the high complexity of the verification problem (non-elementary), we introduce an incomplete procedure for the construction of Presburger stage graphs, and implement it on top of an SMT solver.
The procedure makes extensive use of the theory of well-quasi-orders, and of the structural theory of Petri nets and vector addition systems. We apply our results to a set of benchmarks, in particular to a large collection of population protocols, a model of distributed computation extensively studied by the distributed computing community.

\keywords{parameterized verification \and liveness \and stochastic systems.}
\end{abstract}

\section{Introduction}\label{sec:introduction}
Replicated systems consist of a fully symmetric finite-state program executed by an unknown number of indistinguishable agents, communicating by rendez-vous or via shared variables~\cite{BaslerMWK09,BJKK+15,GS92,KaiserKW10}.
Examples include distributed protocols and multithreaded programs, or abstractions thereof.
The communication graph of replicated systems is a clique. They are a special class of \emph{parameterized systems}, \ie, infinite families of systems that admit a finite description in some suitable modeling language. In the case of replicated systems, the (only) parameter is the number of agents executing the program.

Verifying a replicated system amounts to proving that an infinite family of systems satisfies a given property. This is already a formidable challenge, made even harder by the fact that we want to verify liveness (more difficult than safety) against stochastic schedulers. Loosely speaking, stochastic schedulers select the set of agents that should execute the next action as the result of a random experiment. Stochastic scheduling often appears in distributed protocols, and in particular also in population protocols---a model much studied in distributed computing with applications in computational biology\footnote{Under the name of \emph{chemical reaction networks}.}---that supplies many of our case studies~\cite{AADFP06,NB15}. Under stochastic scheduling, the semantics of a replicated system is an infinite family of finite-state Markov chains. In this work, we study \emph{qualitative} liveness properties, stating that the infinite runs starting at configurations of the system satisfying a precondition almost surely reach and stay in configurations satisfying a postcondition. In this case, whether the property holds or not depends only on the topology of the Markov chains, and not on the concrete probabilities.

We introduce a formal model of replicated systems, based on multiset rewriting, where processes can communicate by shared variables or multiway synchronization. We present a sound and complete verification method called \emph{Presburger stage graphs}. A Presburger stage graphs is a directed acyclic graphs with Presburger formulas as nodes. A formula represents a possibly infinite inductive set of configurations, i.e., a set of configurations closed under reachability. A node $\S$ (which we identify with the set of configurations it represents) has the following property: A run starting at any configuration of $\S$ almost surely reaches some configuration of some successor $\S'$ of $\S$, and, since $\S'$ is inductive, get trapped in $\S'$. A stage graph labels  the node $\S$ with a witness of this property in the form of a \emph{Presburger certificate}, a sort of ranking function expressible in Presburger arithmetic. The completeness of the technique, i.e., the fact that for every property of the replicated system that holds there exists a stage graph proving it, follows from deep results of the theory of vector addition systems (VASs)~\cite{Ler12,Ler13,Ler13b}.

Unfortunately, the theory of VASs also shows that, while the verification problems we consider are decidable, they have non-elementary computational complexity~\cite{CLLLM19}. As a consequence, verification techniques that systematically explore the space of possible stage graphs for a given property are bound to be very inefficient. For this reason, we design an incomplete but efficient algorithm for the computation of stage graphs. Inspired by theoretical results, the algorithm combines a solver for linear constraints with some elements of the theory of well-structured systems~\cite{ACJT96,FS01}. We report on the performance of this algorithm for a large number of case studies. In particular, the algorithm automatically verifies many standard population protocols described in the literature \cite{AGV15,AADFP04,BEJ18,BEJM17,BEK18,CMS10,CDFS11}, as well as liveness properties of distributed algorithms for leader election and mutual exclusion \cite{MutexDijkstra,MutexArray,LeaderHerman,IsraeliJ90,MutexLehmannRabin,MutexBurns,Peterson81,MutexSzymanski}.

\noindent \paragraph{Related work.}   The parameterized verification of replicated systems was first studied in~\cite{GS92}, where they were modeled as counter systems. This allows one to apply many efficient techniques \cite{ALW16,BFHH17,ELMMN14,KKW14}. Most of these works are inherently designed for safety properties, and some can also handle fair termination~\cite{EsparzaM15}, but none of them handles stochastic scheduling. To the best of our knowledge, the only works studying parameterized verification of liveness properties under our notion of stochastic scheduling are those on verification of population protocols. For \emph{fixed} populations,
protocols can be verified with standard probabilistic model checking~\cite{BK08,Var85}, and early works follow this approach~\cite{CMS10,CDFS11,PLD08,SLDP09}. Subsequently, an algorithm
and a tool for the \emph{parameterized} verification of population
protocols were described in~\cite{BEJ18b,BEJM17}, and a first version
of stage graphs was introduced in~\cite{BEK18} for analyzing the
expected termination time of population protocols. In this paper we
overhaul the framework of~\cite{BEK18} for liveness verification,
drawing inspiration from the safety verification technology
of~\cite{BEJ18b,BEJM17}. Compared to~\cite{BEJ18b,BEJM17}, our
approach is not limited to a specific subclass of protocols,
and captures models beyond population protocols. Furthermore, our new
techniques for computing Presburger certificates  subsume
the procedure of~\cite{BEJM17}. In comparison to~\cite{BEK18}, we
provide the first completeness and complexity results for stage
graphs. Further, our stage graphs can prove correctness of population
protocols and even more general liveness properties, while those of~\cite{BEK18} can only
prove termination. We also introduce novel techniques for computing stage graphs,
which compared to~\cite{BEK18} can greatly reduce their size
and allows us to prove more examples correct.

There is also a large body of work on parameterized verification via cutoff techniques: one shows that a specification holds for any number of agents if{}f it holds for any number of agents below some threshold called the cutoff (see~\cite{ARZS15,CGB89,CTTV04,EN03,KaiserKW10}, and~\cite{BJKK+15} for a comprehensive survey). Cut-off techniques can be applied to systems with an array or ring communication structure, but they require the existence and effectiveness of a cutoff, which is not the case in our setting. Further parameterized verification techniques are regular
model checking~\cite{Abd12,BJNT00} and automata learning~\cite{Angluin87}. The classes of communication structures they can handle are orthogonal to ours: arrays and rings for
regular model checking and automata learning, and cliques in our work. Regular model checking and learning have recently been employed to verify safety properties~\cite{ChenHLR17}, liveness properties under arbitrary schedulers~\cite{LR16} and termination under finitary fairness~\cite{LLMR17}. The classes of schedulers considered in~\cite{LLMR17,LR16} are incomparable to ours: arbitrary schedulers in~\cite{LR16}, and finitary-fair schedulers in~\cite{LLMR17}. Further, these works are based on symbolic state-space exploration, while our techniques are based on automatic construction of invariants and ranking functions~\cite{BJKK+15}.

\section{Preliminaries}\label{sec:preliminaries}
Let $\N$ denote $\{0, 1, \ldots\}$ and let $E$ be a finite set. A
\emph{unordered vector} over $E$ is a mapping $V \colon E \to \Z$. In
particular, a \emph{multiset} over $E$ is an unordered vector $M
\colon E \to \N$ where $M(e)$ denotes the number of occurrences of $e$
in $M$. The sets of all unordered vectors and multisets over $E$ are
respectively denoted $\Z^E$ and $\N^E$. Vector addition, subtraction and
comparison are defined componentwise. The \emph{size} of a
multiset $M$ is denoted $\size{M} = \sum_{e \in E} M(e)$.
We let $E^{\ms{k}}$ denote the set of all multisets over $E$ of size $k$.
We sometimes describe multisets using a set-like notation, \eg
$M = \multiset{f, g, g}$ or equivalently $M = \multiset{f, 2 \cdot g}$ is
such that $M(f) = 1$, $M(g) = 2$ and $M(e) = 0$ for all $e \not\in \{f, g\}$.

\paragraph{\Epres{} arithmetic.}\label{sec-presburger-formulas}
Let $X$ be a set of variables.
The set of formulas of \emph{\Epres{} arithmetic} over $X$ is the result of closing atomic formulas, as defined in the next sentence, under Boolean operations and first-order existential quantification. Atomic formulas are of the form
$\sum_{i=1}^k a_i x_i \sim b$, where $a_i$ and $b$ are integers, $x_i$ are variables and
$\sim$ is either $<$ or $\equiv_m$, the latter denoting the congruence modulo $m$ for any $m \ge 2$.
Formulas over $X$ are interpreted on $\N^X$. Given a formula $\phi$ of \Epres{} arithmetic, we let
$\sem{\phi}$ denote the set of all multisets satisfying $\phi$. A set $E \subseteq \N^X$ is a \emph{Presburger set} if $E = \sem{\phi}$ for some formula $\phi$.

\subsection{Replicated systems}
A \emph{replicated system} over $Q$ of arity $n$ is a tuple $\calP = (Q,T)$,
where $T \subseteq \bigcup_{k=0}^n Q^{\ms{k}} \times Q^{\ms{k}}$ is a \emph{transition relation}
containing the set of \emph{silent} transitions $\bigcup_{k=0}^n \{ (\vec{x}, \vec{x}) \mid \vec{x} \in Q^{\ms{k}}) \}$\footnote{
In the paper, we will omit the silent transitions when giving replicated systems.}.
A \emph{configuration} is a multiset $C$ of states, which we interpret as a global state with $C(q)$ agents in each state $q \in Q$.

For every $t = (\vec{x}, \vec{y}) \in T$ with $\vec{x} =
\multiset{X_1, X_2, \ldots, X_k}$ and $\vec{y} = \multiset{Y_1, Y_2,
  \ldots, Y_k}$, we write $X_1 X_2 \cdots X_k \tr Y_1 Y_2 \cdots Y_k$
and let $\pre{t} \defeq \vec{x}$, $\post{t} \defeq \vec{y}$ and $\effect{t} \defeq
\post{t} - \pre{t}$. A transition $t$ is \emph{enabled} at a configuration $C$ if $C \geq \pre{t}$ and, if so, can \emph{occur}, leading to the configuration $C' = C + \effect{t}$. If $t$ is not enabled at $C$, then we say that it is \emph{disabled}. We use the following reachability notation:
\newcommand{\padtrans}[1]{\trans{\mathmakebox[6pt][c]{#1}}}
\begin{align*}
  C \padtrans{t} C'
  &\iff \text{$t$ is enabled at $C$ and its occurrence leads to } C', \\[-3pt]
  C \padtrans{} C'
  &\iff C \trans{t} C' \text{ for some } t \in T, \\[-3pt]
  C \padtrans{w} C' &\iff C = C_0 \trans{w_1}
  C_1 \cdots \trans{w_n} C_n = C' \text{ for some } C_0,
  C_1, \ldots, C_n \in \N^Q, \\[-3pt]
  C \padtrans{*} C' &\iff C \trans{w} C' \text{ for some } w \in T^*.
\end{align*}
Observe that, by definition of transitions, $C \trans{} C'$ implies $\size{C} = \size{C'}$, and likewise for $C \trans{*} C'$.
Intuitively, transitions cannot create or destroy agents.

A \emph{run} is an infinite sequence $C_0 t_1 C_1 t_2 C_2 \cdots$ such
that $C_i \trans{t_{i+1}} C_{i+1}$ for every $i \geq 0$. Given $L \subseteq T^*$ and a set
of configurations $\mathcal{C}$, we let
\begin{align*}
  \postC_L(\C) &\defeq \{C': C \in \C, w \in L, C \trans{w} C'\}, &
  \postC^*(\C) &\defeq \postC_{T^*}(\C), \\
  \preC_L(\C) &\defeq \{C : C' \in \C, w \in L, C \trans{w} C'\}, &
  \preC^*(\C) &\defeq \preC_{T^*}(\C).
\end{align*}

\paragraph{Stochastic scheduling.} We assume that, given a configuration $C$, a probabilistic scheduler picks one of the transitions enabled at $C$. We only make the following two assumptions about the random experiment determining the transition: first, the probability of a transition
depends only on $C$, and, second, every transition enabled at $C$ has a
nonzero probability of occurring. Since $C \trans{*} C'$ implies $\size{C} = \size{C'}$, the number of
configurations reachable from any configuration $C$ is finite. Thus, for every configuration $C$, the semantics of
$\calP$ from $C$ is a finite-state Markov chain rooted at $C$.

\begin{example}
\label{ex:majority}
Consider the replicated system $\calP = (Q,T)$ of arity 2 with states $Q = \{ \AY, \allowbreak \AN, \allowbreak \PY, \PN \}$ and transitions $T = \{ t_1, t_2, t_3, t_4 \}$, where
\begin{alignat*}{5}
    t_1 &\colon & \AY \, \AN &\tr \PY \, \PN, &\qquad\qquad
    t_2 &\colon & \AY \, \PN &\tr \AY \, \PY, \\
    t_3 &\colon & \AN \, \PY &\tr \AN \, \PN, &\qquad\qquad
    t_4 &\colon & \PY \, \PN &\tr \PN \, \PN. \\[-17pt]
\end{alignat*}
\noindent
Intuitively, at every moment in time, agents are either \emph{Active} or \emph{Passive}, and have output \emph{Yes} or \emph{No}, which corresponds to the four states of $Q$.
This system is designed to satisfy the following property: for every configuration $C$ in which all agents are initially active, \ie, $C$ satisfies $C(\PY)=C(\PN)=0$, if $C(\AY) > C(\AN)$, then eventually all agents stay forever in the ``yes'' states $\{\AY, \PY\}$, and otherwise all agents eventually stay forever in the ``no'' states $\{\AN, \PN\}$. \defqed
\end{example}

\subsection{Qualitative model checking}

Let us fix a replicated system $\calP = (Q,T)$. Formulas of \emph{linear temporal logic (LTL)} on $\calP$ are defined by the following grammar:
\[
\varphi ::= \phi \mid \neg \varphi \mid \varphi \lor \varphi \mid \varphi \land \varphi \mid
            \Next \varphi \mid \varphi \Until \varphi
\]
\noindent where $\phi$ is a \Epres{} formula over $Q$. We look at $\phi$ as an atomic proposition over the set $\Nat^Q$ of configurations. Formulas of LTL are interpreted over runs of $\calP$ in the standard way. We abbreviate $\Diamond \varphi \equiv \true \Until \varphi$ and $\Box \varphi \equiv \neg \Diamond \neg \varphi$.

Let us now introduce the probabilistic interpretation of LTL. A configuration $C$ of $\calP$ satisfies an
LTL formula $\varphi$ \emph{with probability $p$} if $\Pr[C, \varphi] = p$, where $\Pr[C, \varphi]$ denotes the probability of the set of runs of $\calP$ starting at $C$ that satisfy $\varphi$ in the finite-state Markov chain rooted at $C$. The measurability of this set of runs for every $C$ and $\varphi$ follows from well-known results~\cite{Var85}.
The \emph{qualitative model checking problem} consists of, given an LTL formula $\varphi$ and a set of configurations $\Init$, deciding whether $\Pr[C , \varphi] =  1$ for every $C \in \Init$.  We will often work with the complement problem, \ie, deciding whether $\Pr[C, \neg\varphi] > 0$ for some \(C\in\Init\).

In contrast to the action-based qualitative model checking problem of~\cite{EsparzaGLM16}, our version of the problem is undecidable due to adding atomic propositions over configurations (see~\Cref{app:preliminaries} for a proof):

\begin{restatable}{theorem}{qualundecidable}\label{thm:qual-undecidable}
  The qualitative model checking problem is not semi-decidable.
\end{restatable}

It is known that qualitative model checking problems of finite-state probabilistic systems reduces to model checking of non-probabilistic systems under an adequate notion of fairness.

\begin{definition}
A run of a replicated system $\calP$ is \emph{fair} if for every possible step $C \trans{t} C'$ of $\calP$ the following holds: if the run contains infinitely many occurrences of $C$, then it also contains infinitely many occurrences of $C \, t \, C'$.
\end{definition}
So, intuitively, if a run can execute a step infinitely often, it eventually will. It is readily seen that a fair run of a finite-state transition system eventually gets ``trapped'' in one of its bottom strongly connected components, and visits each of its states infinitely often. Hence, fair runs of a finite-state Markov chain have probability one. The following proposition was proved in \cite{EsparzaGLM16} for a model slightly less general than replicated systems; the proof can be generalized without effort:

\begin{proposition}[{\cite[Prop.~7]{EsparzaGLM16}}]
Let $\calP$ be a replicated system, let $C$ be a configuration of $\calP$, and let $\varphi$ be an LTL formula. It is the case that $\Pr[C, \varphi] =  1$ if{}f every fair run of $\calP$ starting at $C$ satisfies $\varphi$.
\end{proposition}

We implicitly use this proposition from now on. In particular, we
define:

\begin{definition}
A configuration $C$ \emph{satisfies $\varphi$ with probability 1}, or just \emph{satisfies $\varphi$}, if every fair run starting at $C$ satisfies $\varphi$, denoted by $C \models \varphi$. We let $\sem{\varphi}$ denote the set of configurations satisfying $\varphi$. A set $\C$ of configurations \emph{satisfies} $\varphi$ if $\C \subseteq \sem{\varphi}$, \ie, if $C \models \varphi$ for every $C \in \C$.
\end{definition}

\paragraph{Liveness specifications for replicated systems.}\label{sec-properties}

We focus on a specific class of temporal properties for which the qualitative model checking problem is decidable and which is large enough to formalize many important
specifications. Using well-known automata-theoretic technology, this class can also be used to verify all properties describable in action-based LTL, see e.g. \cite{EsparzaGLM16}.

A \emph{stable termination property} is given by a pair
$\Pi = (\phipre, \Phipost)$, where $\Phipost = \{\phipost^1, \ldots, \phipost^k\}$ and $\phipre, \phipost^1, \ldots, \phipost^k$ are \Epres{} formulas over $Q$ describing sets of configurations.
Whenever $k = 1$, we sometimes simply write $\Pi = (\phipre, \phipost)$. The pair $\Pi$ induces the LTL property
$$
\varphi_\Pi \defeq \Diamond \bigvee_{i=1}^k \Box \phipost^i.
$$
Abusing language, we say that a replicated system $\calP$ \emph{satisfies $\Pi$} if $\sem{\phipre} \subseteq \sem{\varphi_\Pi}$, that is, if every configuration $C$ satisfying $\phipre$ satisfies $\varphi_\Pi$ with probability 1.

\noindent The \emph{stable termination problem} is the qualitative model checking problem for
$\Init = \sem{\phipre}$ and $\varphi = \varphi_\Pi$ given by a stable termination property $\Pi = (\phipre, \Phipost)$.

\begin{example}\label{ex:majority-property}
  Let us reconsider the system from \Cref{ex:majority}. We can
  formally specify that all agents will eventually agree on the
  majority output \emph{Yes} or \emph{No}. Let $\Pi^{\text{Y}} =
  (\phipre^{\text{Y}}, \phipost^{\text{Y}})$ and $\Pi^{\text{N}} =
  (\phipre^{\text{N}}, \phipost^{\text{N}})$ be defined by:
  \begin{align*}
     \phipre^{\text{Y}} &= (\AY > \AN \land \PY + \PN = 0), &
    \phipost^{\text{Y}} &= (\AN + \PN = 0), \\
     \phipre^{\text{N}} &= (\AY \le \AN \land \PY + \PN = 0), &
    \phipost^{\text{N}} &= (\AY + \PY = 0).
  \end{align*}
  The system satisfies the property specified in \Cref{ex:majority} if{}f it satisfies
  $\Pi^{\text{Y}}$ and $\Pi^{\text{N}}$. As an alternative (weaker)
  property, we could specify that the system always stabilizes to
  either output by $\Pi = (\phipre^{\text{Y}} \lor \phipre^{\text{N}},
  \{ \phipost^{\text{Y}}, \phipost^{\text{N}} \})$. \defqed
\end{example}

\section{Stage graphs}\label{sec:stagegraphs}
In the rest of the paper, we fix a replicated system $\calP = (Q,T)$ and a stable termination property $\Pi = (\phipre, \Phipost)$, where $\Phipost = \{\phipost^1, \ldots, \phipost^k\}$, and address the problem of checking whether $\calP$ satisfies $\Pi$. We start with some basic definitions on sets of configurations.

\begin{definition}[inductive sets, leads to, certificates]
\leavevmode
\begin{itemize}
\item A set of configurations $\C$ is \emph{inductive} if $C \in \C$ and $C \rightarrow C'$ implies $C' \in \C$.

\item Let $\C, \C'$ be sets of configurations. We say that
$\C$ \emph{leads to} $\C'$, denoted $\C \leadsto \C'$, if for all $C \in \C$, every fair run from $C$ eventually visits a configuration of~$\C'$.

\item A \emph{certificate} for $\C \leadsto \C'$
is a function $f \colon \C \rightarrow \Nat$ satisfying that for every $C \in \C \setminus \C'$, there exists an execution $C \trans{*} C'$ such that $f(C) > f(C')$.
\end{itemize}
\end{definition}

Note that certificates only require the existence of some executions
decreasing $f$, not for all of them to to decrease it. Despite this,
we have:

\begin{restatable}{proposition}{propCert}\label{prop:cert}
For all inductive sets $\C, \C'$ of configurations, it is the case that: $\C$ leads to $\C'$ if{}f there exists a certificate for $\C \leadsto \C'$.
\end{restatable}

The proof, which can be found in~\Cref{app:stagegraphs}, depends on two properties of replicated systems with stochastic scheduling. First, every configuration has only finitely many descendants. Second, for every fair run and for every finite execution $C \trans{w} C'$, if $C$ appears infinitely often in the run, then the run contains infinitely many occurrences of  $C \trans{w} C'$. We can now introduce stage graphs:

\begin{definition}[stage graph]\label{def:SG}
  A \emph{stage graph} of $\calP$ for the property
  $\Pi$ is a directed
  acyclic graph whose nodes, called \emph{stages}, are sets of
  configurations satisfying the following conditions:
  \begin{enumerate}
  \item every stage is an inductive set;\label{st1}

  \item every configuration of $\sem{\phipre}$ belongs to some
    stage;\label{st2}

  \item if $\C$ is a non-terminal stage with successors $\C_1, \ldots,
    \C_n$, then there exists a certificate for $\C \leadsto (\C_1 \cup
    \cdots \cup \C_n)$;\label{st3}

  \item if $\C$ is a terminal stage, then $\C \models
    \phipost^i$ for some $i$.\label{st4}
  \end{enumerate}
\end{definition}

The existence of a stage graph implies that $\calP$ satisfies $\Pi$. Indeed, by conditions~\ref{st2}--\ref{st3} and repeated application of Proposition~\ref{prop:cert}, every run starting at a configuration of $\sem{\phipre}$ eventually reaches a terminal stage, say $\C$, and, by condition~\ref{st1}, stays in $\C$ forever. Since, by condition~\ref{st4}, all configurations of $\C$ satisfy some $\phipost^i$, after its first visit to $\C$ every configuration satisfies $\phipost^i$.

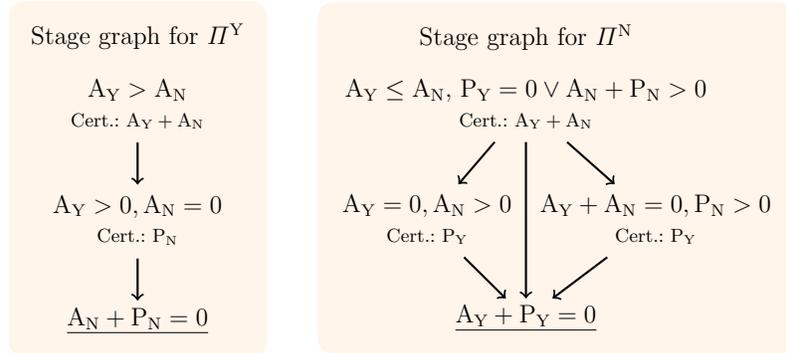
\begin{figure}[t]
\begin{center}
\pgfdeclarelayer{background}
\pgfdeclarelayer{foreground}
\pgfsetlayers{background,main,foreground}
\begin{tikzpicture}[auto, thick, scale=0.8, transform shape, font=\large, xshift=3.6cm, yshift=0cm]

\tikzset{every node/.style={anchor=north}}

  \node[](n5) {\begin{tabular}{c}$\AY > \AN$\\ {\footnotesize Cert.: $\AY+\AN$}\end{tabular}} ;
  \node[](n6) [below =0.7cm of n5] {\begin{tabular}{c}$\AY > 0, \AN=0$\\ {\footnotesize Cert.: $\PN$}\end{tabular}};
  \node[](n7) [below =0.7cm of n6] {$\underline{\AN + \PN=0}$};

  \path[->, thick]
  (n5) edge node[left]{} (n6)
  (n6) edge node[left]{} (n7);

  \node[](b1) [above =0.2cm of n5] {Stage graph for $\Pi^{\text{Y}}$};

  \begin{pgfonlayer}{background}
    \filldraw [line width=4mm,join=round,orange!8]
      (b1.north  -| b1.west)  rectangle (n7.south  -| b1.east);
  \end{pgfonlayer}

  \node[](n0) [right=2.0cm of n5]{\begin{tabular}{c}$\AY \leq \AN$, $\PY = 0 \lor \AN + \PN > 0$\\ {\footnotesize Cert.: $\AY+\AN$}\end{tabular}} ;
  \node[](n1) [below left=0.7cm and -3.1cm of n0, anchor=north east] {\begin{tabular}{c}$\AY = 0, \AN > 0$\\ {\footnotesize Cert.: $\PY$}\end{tabular}};
  \node[](n2) [below right=0.7cm and -3.1cm of n0, anchor=north west] {\begin{tabular}{c} $\AY + \AN=0, \PN > 0$\\ {\footnotesize Cert.: $\PY$}\end{tabular}};
  \node[](n3) [below =2.6cm of n0] {$\underline{\AY + \PY=0}$};

  \node[](inter) [below=1.5cm of n0] {};

  \path[->, thick]
  (n0) edge node[above]{} (n1)
  (n0) edge node[above]{} (n2)
  (n1) edge node[left]{}  (n3)
  (n2) edge node[left]{}  (n3)
  (n0) edge node[left]{}  (n3);

  \node[](b0) [above =0.2cm of n0] {Stage graph for $\Pi^{\text{N}}$};
  \begin{pgfonlayer}{background}
    \filldraw [line width=4mm,join=round,orange!8]
      (b0.north  -| n1.west)  rectangle (n3.south  -| n2.east);
  \end{pgfonlayer}

\end{tikzpicture}
\end{center}
\caption{Stage graphs for the system of \Cref{ex:majority}.}
\label{fig-stage-graph}
\end{figure}

\begin{example}
\Cref{fig-stage-graph} depicts stage graphs for the system of \Cref{ex:majority} and the properties defined in \Cref{ex:majority-property}.
The reader can easily show that every stage $\C$ is inductive by checking that for every $C \in \C$ and every transition $t \in \{t_1, \ldots, t_4\}$ enabled at $C$, the step
$C \trans{t_i} C'$ satisfies $C' \in \C$. For example, if a configuration satisfies $\AY > \AN$, so does any successor configuration. \defqed
\end{example}

The following proposition shows that stage graphs are a sound and complete technique for proving stable termination properties.

\begin{restatable}{proposition}{propComp}\label{propo:comp}
  System $\calP$ satisfies $\Pi$ if{}f it has a stage graph for $\Pi$.
\end{restatable}

Proposition \ref{propo:comp} does not tell us anything about the decidability of the stable termination problem. To prove that the problem is decidable, we introduce \Epres{} stage graphs. Intuitively these are stage graphs whose stages and certificates can be expressed by formulas of \Epres{} arithmetic.

\begin{definition}[Presburger stage graphs]
\leavevmode
\begin{itemize}
    \item A stage $\C$ is \emph{\Epres{}} if $\C = \sem{\phi}$ for some \Epres{} formula $\phi$.

    \item A \emph{bounded certificate} for $\C \leadsto \C'$ is a pair $(f,k)$, where $f \colon \C \to \Nat$ and $k \in \N$, satisfying that for every $C \in \C \setminus \C'$, there exists an execution $C \trans{w} C'$ such that $f(C) > f(C')$ and $\size{w} \le k$.

    \item A \emph{\Epres{} certificate} is a bounded certificate $(f,k)$ satisfying $f(C)= n \iff \varphi(C,n)$ for some \Epres{} formula $\varphi(\vec{x}, y)$.
    \item A \emph{\Epres{} stage graph} is a stage graph whose  stages and certificates are all \Epres{}.
\end{itemize}
\end{definition}

Using a powerful result from~\cite{EGLM17}, we show that: (1) $\calP$ satisfies $\Pi$ if{}f it has a \Epres{} stage graph for $\Pi$ (\Cref{thm:stagegraph}); (2) there exists a denumerable set of candidates for a \Epres{} stage graph for $\Pi$; and (3) there is an algorithm that decides whether a given candidate is a \Epres{} stage graph for $\Pi$ (\Cref{thm:stagegraphdec}). Together, (1--3) show that the stable termination problem is semi-decidable. To obtain decidability, we observe that the complement of the stable termination problem is also semi-decidable. Indeed, it suffices to enumerate all initial configurations $C \models \phipre$, build for each such $C$ the (finite) graph $G_C$ of configurations reachable from $C$, and check if some bottom strongly connected component
$\Bottom$ of $G_C$ satisfies $\Bottom \not\models \phipost^i$ for all $i$. This is the case if{}f some fair run starting at $C$ visits and stays in $\Bottom$, which in turn is the case if{}f $\calP$ violates~$\Pi$.

\begin{restatable}{theorem}{thmStageGraph}\label{thm:stagegraph}
 System $\calP$ satisfies $\Pi$ if{}f it has a \Epres{} stage graph for $\Pi$.
\end{restatable}

We observe that testing whether a given graph is a \Epres{} stage
graph reduces to Presburger arithmetic satisfiability, which is
decidable~\cite{Presburger} and whose complexity lies between
\textsf{2-}\NEXP{} and \textsf{2-}\EXPSPACE~\cite{Ber80}:

\begin{restatable}{theorem}{thmStageGraphDec}\label{thm:stagegraphdec}
The problem of deciding whether an acyclic graph of Presburger sets and Presburger certificates  is a \Epres{} stage graph, for a given stable termination property, is reducible in polynomial time to the satisfiability problem for Presburger arithmetic.
\end{restatable}

\section{Algorithmic construction of stage graphs}\label{sec:stagegraph_construction}
At the current state of our knowledge, the decision procedure derived from Theorem \ref{thm:stagegraphdec} has little practical relevance. From a theoretical point of view, the \TOWER-hardness result of~\cite{CLLLM19} implies that the stage graph may have non-elementary size in the system size. In practice, systems have relatively small stage graphs, but, even so, the enumeration of all candidates immediately leads to a prohibitive combinatorial explosion.

For this reason, we present a procedure to automatically \emph{construct} (not guess) a Presburger stage graph $G$ for a given replicated system $\calP$ and a stable termination property
$\Pi = (\phipre, \Phipost)$. The procedure may \emph{fail}, but, as shown in the experimental section, it succeeds for many systems from the literature.

The procedure is designed to be implemented on top of a solver for the existential fragment of Presburger arithmetic. While every formula of Presburger arithmetic has an equivalent formula within the existential fragment~\cite{Cooper72,Presburger}, quantifier-elimination may lead to a doubly-exponential blow-up in the size of the formula. Thus, it is important to emphasize that our procedure \emph{never requires to eliminate quantifiers}: If the pre- and postconditions of $\Pi$ are supplied as quantifier-free formulas, then all constraints of the procedure remain in the existential fragment.

We give a high-level view of the procedure (see Algorithm \ref{alg:sg}), which uses several functions, described in detail in the rest of the paper. The procedure  maintains a workset $\WS$ of Presburger stages, represented by existential Presburger formulas. Initially, the only stage is an inductive Presburger overapproximation $\PReach(\sem{\phipre})$ of the configurations
reachable from $\sem{\phipre}$ ($\PReach$ is an abbreviation for ``potentially reachable''). Notice that we must necessarily use an overapproximation, since $\postC^*(\sem{\phipre})$ is not always
expressible in Presburger arithmetic\footnote{This follows easily from the fact that $\postC^*(\psi)$ is not always expressible in Presburger arithmetic for vector addition systems, even if $\psi$ denotes a single configuration~\cite{HP79}.}.
We use a refinement of the overapproximation introduced in~\cite{BEJM17,ELMMN14}, equivalent to
the overapproximation of~\cite{BFHH17}.

\begin{algorithm}[t]
  \DontPrintSemicolon
  \LinesNumbered
  \KwIn{replicated system $\PP = (Q, T)$, stable term.\ property $\Pi = (\phipre, \Phipost)$}
  \KwResult{a stage graph of $\PP$ for $\Pi$}
  \algskip

  $\WS \leftarrow \{ \PReach(\sem{\phipre}) \}$\;
  \algskip

  \While{$\WS \neq \emptyset$}{
    \textbf{remove} $\Stage$ \textbf{from} $\WS$\;
    \algskip

    \If{$\neg \Terminal(\Stage, \Phipost)$}{
      $U \leftarrow \EventDead(\Stage)$\;
      \algskip

      \If{$U \neq \emptyset$}{
        $\WS \leftarrow \WS \cup \{\IndOverapprox(\Stage, U)\}$
      }
      \Else{
        $\WS \leftarrow \WS \cup \Split(\Stage)$
      }
    }
  }
  \caption{procedure for the construction of stage graphs.}\label{alg:sg}
\end{algorithm}

In its main loop (lines 2--9),  Algorithm \ref{alg:sg} picks a Presburger stage $\Stage$ from the workset, and processes it. First, it calls $\Terminal(\S,\Phipost)$ to check if $\Stage$ is terminal, \ie, whether $\Stage \models \phipost^i$ for some $\phipost^i \in \Phipost$. This reduces to checking the unsatisfiability of the existential Presburger formula $\phi \land \neg \phipost^i$, where $\phi$ is the formula characterizing $\Stage$. If $\S$ is not terminal, then the procedure attempts to construct successor stages in lines 5--9, with the help of three further functions: $\EventDead$,
$\IndOverapprox$, and $\Split$. In the rest of this section, we present the intuition behind lines 5--9,
and the specification of the three functions. \Cref{sec:setconfdead,sec:split,sec:evdead} present the implementations we use for these functions.

Lines 5--9 are inspired by the behavior of most replicated systems designed by humans, and are based on the notion of \emph{dead} transitions,
which can never occur again (to be formally defined below).
Replicated systems are usually designed to run in \emph{phases}. Initially, all transitions are alive, and the end of a phase is marked by the ``death'' of one or more transitions, i.e., by reaching a configuration at which these transitions are dead. The system keeps ``killing transitions'' until no
transition that is still alive can lead to a configuration violating the postcondition. The procedure mimics this pattern. It constructs stage graphs in which if $\Stage'$ is a successor of $\Stage$, then the set of transitions dead at $\Stage'$ is a \emph{proper superset} of the transitions dead at $\Stage$. For this, $\EventDead(\Stage)$  computes a set of transitions that are alive at some configuration of $\Stage$, but which will become dead in every fair run starting at $\Stage$ (line~5). Formally, $\EventDead(\Stage)$ returns a set $U \subseteq \overline{\Dead(\Stage)}$ such that $\Stage \models \Diamond \dead{U}$, defined as follows.

\begin{definition}\label{def:dis-dead}
A transition of a replicated system $\calP$ is \emph{dead} at a configuration $C$ if it is disabled at every configuration reachable from $C$ (including $C$ itself). A transition is \emph{dead} at a stage $\Stage$ if it is dead at every configuration of $\Stage$.
Given a stage $\Stage$ and a set $U$ of transitions, we use the following notations:
\begin{itemize}
 \item $\Dead(\Stage)$: the set of transitions dead at $\Stage$;
 \item $\sem{\disa{U}}$: the set of configurations at which all transitions of $U$ are disabled;
 \item $\sem{\dead{U}}$: the set of configurations at which all transitions of $U$ are dead.
\end{itemize}
\end{definition}

Observe that we can compute $\Dead(\Stage)$ by checking unsatisfiability of a sequence of existential Presburger formulas:
as $\Stage$ is inductive, we have $\Dead(\Stage) = \{ t \mid \Stage \models \disa{t} \}$, and $\Stage \models \disa{t}$ holds
if{}f the existential Presburger formula $\exists C \colon \phi(C) \land C \ge \pre{t}$ is unsatisfiable, where $\phi$ is the formula characterizing $\Stage$.

The following proposition, whose proof appears in~\Cref{app:stagegraph_construction}, shows that determining whether a given transition will eventually become dead, while decidable, is \PSPACE-hard.
Therefore, \Cref{sec:evdead} describes two implementations of this function, and a way to combine them, which exhibit a good trade-off between precision and computation time.
\begin{restatable}{proposition}{eventdeadhardness}\label{prop:eventdead-hardness}
    Given a replicated system $\calP$, a stage $\Stage$ represented by an existential Presburger formula $\phi$ and a set of transitions $U$,
    determining whether $\Stage \models \Diamond \dead{U}$ holds is decidable and \PSPACE-hard.
\end{restatable}

If the set $U$ returned by $\EventDead(\Stage)$ is nonempty, then we know that every fair run starting at a configuration of $\Stage$ will eventually reach a configuration of $\Stage \cap \sem{\dead{U}}$. So, this set, or any inductive overapproximation of it, can be a legal successor of $\Stage$ in the stage graph. Function $\IndOverapprox(\Stage,U)$ returns such an inductive  overapproximation  (line 7). To be precise, we show in \Cref{sec:setconfdead} that $\sem{\dead{U}}$ is a Presburger set that can be computed exactly, albeit in doubly-exponential time in the worst case. The section also shows how to compute overapproximations more efficiently.
If the set $U$ returned by $\EventDead(\Stage)$ is empty, then we cannot yet construct any successor of $\Stage$. Indeed, recall that we want to construct stage graphs in which if $\Stage'$ is a successor of $\Stage$, then $\Dead(\Stage')$ is a \emph{proper superset} of $\Dead(\Stage)$. In this case,
we proceed differently and try to split $\Stage$:

\begin{definition}
A \emph{split} of some stage $\Stage$ is a set $\{\Stage_1, \ldots, \Stage_k\}$ of (not necessarily disjoint) stages such that the following holds:
	\begin{itemize}
        \item  $\Dead(\Stage_i) \supset \Dead(\Stage)$ for every $1 \leq i \leq k$, and
        \item  $\Stage = \bigcup_{i=1}^k \Stage_i$.
	\end{itemize}
\end{definition}

If there exists a split $\{\Stage_1, \ldots, \Stage_k\}$ of $\Stage$, then we can let
$\Stage_1, \ldots, \Stage_k$ be the successors of $\Stage$ in the stage graph.
Observe that a stage may indeed have a split. We have $\Dead(\C_1 \cup \C_2) = \Dead(\C_1) \cap \Dead(\C_2)$, and hence $\Dead(\C_1 \cup \C_2)$ may be a proper subset of both $\Dead(\C_1)$ and $\Dead(\C_2)$:

\begin{example}
	Consider the system with states $\{q_1, q_2\}$ and transitions
	$t_i \colon q_i \mapsto q_i$ for $i \in \{1, 2\}$. Let
	$\Stage = \{ C \mid C(q_1) = 0 \vee C(q_2) = 0 \}$, \ie, $\S$ is
	the (inductive) stage of configurations disabling either $t_1$ or
	$t_2$. The set $\{ \Stage_1, \Stage_2 \}$, where $\Stage_i = \{ C
	\in \Stage \mid C(q_i) = 0 \}$, is a split of $\Stage$ satisfying
	$\Dead(\Stage_i) = \{t_i\} \supset \emptyset = \Dead(\Stage)$. \defqed
\end{example}

\noindent The canonical split of $\Stage$, if it exists, is the set $\{ \Stage \cap \sem{\dead{t}} \mid t \notin\Dead(\Stage) \}$. As mentioned above, \Cref{sec:setconfdead} shows that $\sem{\dead{U}}$ can be computed exactly for every $U$, but the computation can be expensive.
Hence, the canonical split can be computed exactly at potentially high cost. Our implementation uses an underapproximation of $\sem{\dead{t}}$, described in \Cref{sec:split}.

\section[Computing and approximating dead(U)]{Computing and approximating $\sem{\dead{U}}$} \label{sec:setconfdead}
We show that, given a set $U$ of transitions,
\begin{itemize}
	\item we can effectively compute an existential Presburger formula describing the set $\sem{\dead{U}}$, with high computational cost in the worst case, and
	\item we can effectively compute constraints that overapproximate or underapproximate $\sem{\dead{U}}$, at a reduced computational cost.
\end{itemize}
\subsubsection{Downward and upward closed sets.}

We enrich $\Nat$ with the limit element $\omega$ in the usual way. In
particular, $n < \omega$ holds for every $n \in \Nat$. An
\emph{$\omega$-configuration} is a mapping $C^\omega \colon Q
\rightarrow \Nat \cup \{\omega\}$. The \emph{upward closure} and
\emph{downward closure} of a set $\C^\omega$ of
$\omega$-configurations are the sets of configurations $\Up{\C^\omega}$ and $\Down{\C^\omega}$, respectively defined as:
\begin{align*}
  \Up{\C^\omega} &\defeq \{C \in \N^Q \mid C \geq C^\omega \text{ for some }
  C^\omega \in \C^\omega\}, \\
  \Down{\C^\omega} &\defeq \{C \in \N^Q \mid C \leq C^\omega \text{ for some }
  C^\omega \in \C^\omega\}.
\end{align*}
A set $\C$ of configurations is \emph{upward closed} if
$\C = \Up{\C}$, and \emph{downward closed} if $\C = \Down{\C}$. These
facts are well-known from the theory of well-quasi orderings:

\begin{lemma}\label{lem:updown}
  For every set $\C$ of configurations, the following holds:
  \begin{enumerate}
  \item $\C$ is upward closed iff $\overline{\C}$ is downward closed
    (and vice versa);

  \item if $\C$ is upward closed, then there is a unique
    minimal finite set of configurations $\inf{\C}$, called its
    \emph{basis}, such that $\C = \Up{\inf{\C}}$;

  \item if $\C$ is downward closed, then there is a unique minimal
    finite set of $\omega$-configurations $\sup{\C}$, called its
    \emph{decomposition}, such that $\C = \Down{\sup{\C}}$.
  \end{enumerate}
\end{lemma}

\subsubsection{Computing $\sem{\dead{U}}$ exactly. }

It follows immediately from \Cref{def:dis-dead} that both $\sem{\disa{U}}$ and $\sem{\dead{U}}$ are downward closed. Indeed, if all transitions of $U$ are disabled at $C$, and $C' \leq C$, then they are also disabled at $C'$, and clearly the same holds for transitions dead at $C$. Furthermore:

\begin{proposition}\label{prop:dis-and-dead}
  For every set $U$ of transitions, the (downward) decomposition of
  both $\sup{\sem{\disa{U}}}$ and $\sup{\sem{\dead{U}}}$ is
  effectively computable.
\end{proposition}

\begin{proof}
	For every $t \in U$ and $q \in \pre{t}$, let $C_{t, q}^\omega$
	be the $\omega$-configuration such that $C_{t, q}^\omega(q) =
	\pre{t}(q) - 1$ and $C_{t, q}^\omega(p) = \omega$ for every $p \in Q
	\setminus \{q\}$. In other words, $C_{t, q}^\omega$ is the
	$\omega$-configuration made only of $\omega$'s except for state $q$
	which falls short from $\pre{t}(q)$ by one. This
	$\omega$-configurations captures all configurations disabled in $t$
	due to an insufficient amount of agents in state $q$. We have:
	$$\sup{\sem{\disa{U}}} = \{C_{t, q}^\omega : t \in U, q \in
        \pre{t}\}.$$ The latter can be made minimal by removing
        superfluous $\omega$-configurations.

	For the case of $\sup{\sem{\dead{U}}}$, we invoke~\cite[Prop.~2]{JP19} which gives a proof for the more general setting of
	(possibly unbounded) Petri nets. Their procedure is based on the
	well-known backwards reachability algorithm (see, e.g., \cite{ACJT96,FS01}). \qed
\end{proof}

Since $\sup{\sem{\dead{U}}}$ is finite, its computation allows to
describe $\sem{\dead{U}}$ by the following linear
constraint\footnote{Observe that if $C^\omega(q) = \omega$, then the
  term ``$C(q) \leq \omega$'' is equivalent to ``$\mathbf{true}$''.}:
$$\bigvee_{C^\omega \in \sup{\sem{\dead{U}}}} \bigwedge_{q \in Q}
\left[C(q) \leq C^\omega(q)\right].$$ However, the cardinality of
$\sup{\sem{\dead{U}}}$ can be exponential~\cite[Remark for Prop.~2]{JP19} in the
system size. For this reason, we are interested in constructing both
under- and over-approximations.

\subsubsection{Overapproximations of $\sem{\dead{U}}$.}
For every $i \in \N$, define $\sem{\dead{U}}^i$ as:
$$
\sem{\dead{U}}^0 \defeq \sem{\disa{U}}
\quad \text{ and } \quad
\sem{\dead{U}}^{i+1} \defeq \overline{\preC_T(\overline{\sem{\dead{U}}^{i}})} \cap \sem{\disa{U}}.
$$
Loosely speaking, $\sem{\dead{U}}^{i}$ is the set of configurations
$C$ such that every configuration reachable in at most
$i$ steps from $C$ disables $U$. We immediately have:
$$
\sem{\dead{U}} = \bigcap_{i=0}^\infty \sem{\dead{U}}^{i}.
$$
Using~\Cref{prop:dis-and-dead} and the following proposition, we obtain that $\sem{\dead{U}}^{i}$ is an effectively computable overapproximation of $\sem{\dead{U}}$.

\begin{restatable}{proposition}{postpresburger}\label{prop:post-presburger}
    For every Presburger set $\C$ and every set of transitions $U$, the sets $\preC_U(\C)$ and $\postC_U(\C)$ are effectively Presburger.
\end{restatable}

Recall that function $\IndOverapprox(\Stage, U)$ of Algorithm~\ref{alg:sg} must return
an \emph{inductive} overapproximation of $\sem{\dead{U}}$.  Since $\sem{\dead{U}}^i$ might not be inductive in general, our implementation uses either the inductive overapproximations  $\IndOverapprox^i(\Stage, U) \defeq \PReach(\Stage \cap \sem{\dead{U}}^i)$, or the exact value $\IndOverapprox^\infty(\Stage, U) \defeq \Stage \cap \sem{\dead{U}}$. The table of results in the experimental section describes for each benchmark which overapproximation was used.

\subsubsection{Underapproximations of $\sem{\dead{U}}$: Death certificates.}
A \emph{death certificate} for $U$ in $\calP$ is a finite set $\C^\omega$
of $\omega$-configurations such that:
\begin{enumerate}
	\item $\Down{\C^\omega} \models \disa{U}$, i.e., every configuration of
	$\Down{\C^\omega}$ disables $U$, and
	\item $\Down{\C^\omega}$ is inductive, i.e., $\postC_T(\Down{\C^\omega}) \subseteq \Down{\C^\omega}$.
\end{enumerate}

If $U$ is dead at a set $\C$ of configurations, then there is always a certificate that proves it,
namely $\sup{\sem{\dead{U}}}$. In particular, if $\C^\omega$ is a death certificate for $U$
then $\Down{\C^\omega} \subseteq \sem{\dead{U}}$, that is, $\Down{\C^\omega}$ is
an underapproximation of $\sem{\dead{U}}$

Using \Cref{prop:post-presburger}, it is straightforward
to express in \Epres{} arithmetic that a finite set $\C^\omega$ of $\omega$-configurations is a death certificate for $U$:

\begin{restatable}{proposition}{deathcert}\label{prop:death-cert}
For every $k \geq 1$ there is an existential Presburger formula $\DeathCert_k(U, \C^\omega)$ that holds if{}f  $\C^\omega$ is a death certificate of size $k$ for $U$.
\end{restatable}

\section{Splitting a stage} \label{sec:split}
Given a stage $\Stage$, we try to find a set $\C^\omega_1, \ldots, \C^\omega_\ell$ of death certificates for transitions $t_1, \ldots, t_\ell \in T \setminus \Dead(\Stage)$ such that
$\Stage \subseteq \Down{\C^\omega_1} \cup \cdots \cup \Down{\C^\omega_\ell}$. This allows us to split $\Stage$ into $\Stage_1, \ldots, \Stage_\ell$, where $\Stage_i \defeq \Stage \cap \Down{\C^\omega_i}$.

For any fixed size $k \geq 1$ and any fixed $\ell$, we can find death certificates $\C^\omega_1, \ldots, \C^\omega_\ell$ of size at most $k$ by solving a Presburger formula. However, the formula does not belong to the existential fragment, because the inclusion check $\Stage \subseteq \Down{\C^\omega_1} \cup \cdots \cup \Down{\C^\omega_\ell}$ requires universal quantification. For this reason, we proceed iteratively. For every $i \geq 0$,
after having found $\C^\omega_1, \ldots, \C^\omega_i$ we search for a pair $(C_{i+1}, \C^\omega_{i+1})$ such that
\begin{enumerate}[label={(\roman*)}]
\item\label{itm:split1} $\C^\omega_{i+1}$ is a death certificate for some $t_{i+1} \in T \setminus \Dead(\Stage)$;
\item\label{itm:split2} $C_{i+1} \in \Stage \cap \Down{\C^\omega_{i+1}} \setminus (\Down{\C^\omega_1} \cup \cdots \cup \Down{\C^\omega_i})$.
\end{enumerate}
An efficient implementation requires to guide the search for $(C_{i+1}, \C^\omega_{i+1})$, because otherwise the search procedure might not even terminate, or might split $\Stage$ into too many parts, blowing up the size of the stage graph. Our search procedure employs the following heuristic, which works well in practice. We only consider the case $k=1$, and search for a pair $(C_{i+1}, C^\omega_{i+1})$ satisfying \ref{itm:split1} and \ref{itm:split2} above, and additionally:
\begin{enumerate}[resume,label={(\roman*)}]
\item\label{itm:split3} all components of $C^\omega_{i+1}$ are either $\omega$ or between $0$ and $\max_{t \in T,q \in Q} \pre{t}(q)-1$;
\item\label{itm:split4} for every $\omega$-configuration $C^ \omega$, if $(C_{i+1}, C^\omega)$ satisfies \ref{itm:split1}--\ref{itm:split3}, then $C^\omega_{i+1} \leq C^\omega$;
\item\label{itm:split5} for every pair $(C, C^\omega)$, if $(C, C^\omega)$ satisfies \ref{itm:split1}--\ref{itm:split4}, then $C^\omega \leq C^\omega_{i+1}$.
\end{enumerate}
\noindent Condition~\ref{itm:split3} guarantees termination.
Intuitively, condition~\ref{itm:split4} leads to certificates valid for sets $
U \subseteq T \setminus \Dead(\Stage)$ as large as possible. So it allows us to avoid splits that, loosely speaking, do not make as much progress as they could. Condition \ref{itm:split5} allows us to avoid splits with many elements because each element of the split has a small intersection with $\Stage$.

An example illustrating these conditions is given in~\Cref{app:split}.

\section{Computing eventually dead transitions}\label{sec:evdead}

Recall that the function $\EventDead(\Stage)$ takes an inductive Presburger set $\Stage$ as input, and returns a (possibly empty) set $U \subseteq \overline{\Dead(\Stage)}$ of transitions such that $\Stage \models \Diamond \dead{U}$. This guarantees $\Stage \leadsto \sem{\dead{U}}$ and, since $\Stage$ is inductive, also $\Stage \leadsto \Stage \cap \sem{\dead{U}}$.

By \Cref{prop:eventdead-hardness}, deciding if there exists a non-empty set $U$ of transitions such that $\Stage \models \Diamond \dead{U}$ holds is \PSPACE-hard, which makes a polynomial reduction to satisfiability of existential Presburger formulas unlikely.
So we design incomplete implementations of $\EventDead(\Stage)$ with lower complexity. Combining these implementations, the lack of completeness essentially vanishes in practice.

The implementations are inspired by Proposition~\ref{prop:cert}, which shows that $\Stage\leadsto \sem{\dead{U}}$ holds if{}f there exists a certificate $f$ such that:
\begin{equation}
\tag{Cert}
\forall C \in \Stage \setminus \sem{\dead{U}}\ \colon \exists \, C \trans{*} C' \colon f(C) > f(C').
\label{cond:cert}
\end{equation}
To find such certificates efficiently, we only search for \emph{linear} functions $f(C) = \sum_{q \in Q} \vec{a}(q) \cdot C(q)$ with coefficients $\vec{a}(q) \in \N$ for each $q \in Q$.

\subsection{First implementation: Linear ranking functions}\label{asdead-invariants}

Our first procedure computes the existence of a linear \emph{ranking function}.

\begin{definition}\label{def:ranking}
    A function $r \colon \Stage \rightarrow \Nat$ is a ranking function for $\Stage$ and $U$ if for every $C \in \Stage$ and every step $C \trans{t} C'$ the following holds:
    \begin{enumerate}
        \item\label{itm:ranking1} if $t \in U$, then $r(C) > r(C')$; and
        \item\label{itm:ranking2} if $t \notin U$, then $r(C) \geq r(C')$.
    \end{enumerate}
\end{definition}

\begin{proposition}\label{prop:ranking}
  If $r \colon \Stage \to \Nat$ is a ranking function for $\Stage$ and
  $U$, then there exists $k \in \N$ such that $(r, k)$ is a bounded
  certificate for $\Stage \leadsto \sem{\dead{U}}$.
\end{proposition}

\begin{proof}
  Let $M$ be the minimal finite basis of the upward closed set
  $\overline{\sem{\dead{U}}}$. For every configuration $D \in M$, let
  $\sigma_D$ be a shortest sequence that enables some transition of
  $t_D \in U$ from $D$, \ie, such that $D \trans{\sigma_D} D' \trans{t_D} D''$
  for some $D'$, $D''$. Let $k \defeq \max\{|\sigma_D t_D|
  : D \in M\}$.

  Let $C \in \Stage \setminus \sem{\dead{U}}$. Since $C \in
  \overline{\sem{\dead{U}}}$, we have $C \ge D$ for some $D \in M$. By
  monotonicity, we have $C \trans{\sigma_D} C' \trans{t_D} C''$ for
  some configurations $C'$ and $C''$. By Definition~\ref{def:ranking},
  we have $r(C) \geq r(C') > r(C'')$, and so
  condition~\eqref{cond:cert} holds. As $\size{\sigma_D t_D} \le k$, we
  have that $(r, k)$ is a bounded certificate. \qed
\end{proof}

It follows immediately from Definition~\ref{def:ranking} that if $r_1$
and $r_2$ are ranking functions for sets $U_1$ and $U_2$ respectively,
then $r$ defined as $r(C) \defeq r_1(C) + r_2(C)$ is a ranking function for
$U_1 \cup U_2$. Therefore, there exists a unique maximal set of
transitions $U$ such that $\Stage \leadsto \sem{\dead{U}}$ can be
proved by means of a ranking function. Further, $U$ can be computed
by collecting all transitions $t \in \overline{\Dead(\Stage)}$
such that there exists a ranking function $r_t$ for $\{t\}$. The existence
of a \emph{linear} ranking function $r_t$ can be decided in
polynomial time via linear programming, as follows. Recall that for every step $C
\trans{u} C'$, we have $C' = C + \effect{u}$. So, by linearity, we
have $r_t(C) \geq r_t(C') \iff r_t(C' - C) \leq 0 \iff r_t(\effect{u})
\leq 0$. Thus, the constraints of Definition~\ref{def:ranking} can be
specified as:
\begin{align*}
  \vec{a} \cdot \effect{t} < 0  \quad \land
  \bigwedge_{u \in \overline{\Dead(\Stage)}} \vec{a} \cdot
  \effect{u} \le 0 ,
\end{align*}
where $\vec{a} \colon Q \rightarrow \Q_{\geq 0}$ gives the coefficients of $r_t$, that is,
$r_t(C) = \vec{a} \cdot C$, and
$\vec{a} \cdot \vec{x} \defeq \sum_{q \in Q} \vec{a}(q) \cdot \vec{x}(q)$ for $\vec{x} \in \N^Q$.
Observe that a solution may yield a
function whose codomain differs from $\N$. However, this is not an
issue since we can scale it with the least common denominator of each
$\vec{a}(q)$.

\subsection{Second implementation: Layers}\label{asdead-layers}

\emph{Transitions layers} were introduced in~\cite{BEJM17} as a technique to find transitions that will eventually become dead. Intuitively, a set $U$ of transitions is a layer if (1) no run can contain only transitions of $U$, and (2) $U$ becomes dead once disabled; the first condition guarantees that $U$ eventually becomes disabled, and the second that it eventually becomes dead.
We formalize layers in terms of \emph{layer functions}.

\begin{definition}\label{def:layer}
    A function $\ell \colon \Stage \to \Nat$ is a \emph{layer function} for $\Stage$ and $U$ if:
    \begin{enumerate}[leftmargin=2.4em,label={\textbf{C\arabic*}.},ref={\textbf{C\arabic*}}]
        \item\label{itm:layer1}  $\ell(C) > \ell(C')$  for every $C \in \Stage$ and every step $C \trans{t} C'$ with $t \in U$; and
        \item\label{itm:layer2}  $\sem{\disa{U}} = \sem{\dead{U}}$.
    \end{enumerate}
\end{definition}

\begin{proposition}
  If $\ell \colon \Stage \to \Nat$ is a layer function for $\Stage$ and
  $U$, then $(\ell, 1)$ is a bounded certificate for $\Stage \leadsto
  \sem{\dead{U}}$.
\end{proposition}

\begin{proof}
  Let $C \in \Stage \setminus \sem{\dead{U}}$.
  By condition~\ref{itm:layer2}, we have $C \not\in \sem{\disa{U}}$.
  So there exists a step $C \trans{u} C'$ where $u \in U$.
  By condition~\ref{itm:layer1}, we have $\ell(C) > \ell(C')$, so
  condition~\eqref{cond:cert} holds and $(\ell, 1)$ is a bounded certificate.
\end{proof}

Let $\Stage$ be a stage. For every set of transitions $U \subseteq \overline{\Dead(\Stage)}$ we can construct a Presburger formula $\linlayer(U, \vec{a})$ that holds if{}f there there exists a \emph{linear} layer function for $U$, i.e., a layer function of the form $\ell(C) = \vec{a} \cdot C$ for a vector of coefficients $\vec{a} \colon Q \rightarrow  \Q_{\geq 0}$.
Condition~\ref{itm:layer1}, for a linear function $\ell(C)$, is expressed by the existential Presburger formula
\begin{equation*}\label{constraint-cond-1}
\linlayerfun(U, \vec{a}) \defeq \bigwedge_{u \in U} \vec{a} \cdot \effect{u} < 0.
\end{equation*}
Condition~\ref{itm:layer2} is expressible in Presburger arithmetic because of \Cref{prop:dis-and-dead}.
However, instead of computing $\sem{\dead{U}}$ explicitly, there is a more efficient way to express this constraint.
Intuitively, $\sem{\disa{U}} = \sem{\dead{U}}$ is the case if enabling a transition $u \in U$ requires to have previously enabled some transition $u' \in U$. This observation leads to:
\begin{restatable}{proposition}{layerdisabled}\label{prop:layer-disabled}
A set $U$ of transitions satisfies $\sem{\disa{U}} = \sem{\dead{U}}$ if{}f it satisfies the existential Presburger formula
    \begin{align*}
      \diseqdead(U) \defeq \bigwedge_{t \in T} \bigwedge_{u \in U} \bigvee_{u' \in U} \pre{t} + (\pre{u} \mminus \post{t}) \ge \pre{u'}
    \end{align*}
    where $\vec{x} \mminus \vec{y} \in \N^Q$ is defined by $(\vec{x} \mminus \vec{y})(q) \defeq \max(\vec{x}(q) - \vec{y}(q), 0)$ for $\vec{x},\vec{y} \in \N^Q$.
\end{restatable}
\noindent This allows us to give the constraint $\linlayer(U, \vec{a})$, which is of polynomial size:
\begin{equation*}
    \linlayer(U, \vec{a}) \defeq \linlayerfun(U, \vec{a})  \wedge \diseqdead(U).
\end{equation*}

\subsection{Comparing ranking and layer functions}\label{asdead-combined}

The ranking and layer functions of Sections~\ref{asdead-invariants} and~\ref{asdead-layers} are  incomparable in power, that is,
there are sets of transitions for which a ranking function but no
layer function exists, and vice versa.  This is shown by the following two
systems:
\begin{alignat*}{5}
\calP_1 &= (\ \{\ \text{A}, \text{B}, \text{C}\ \},&\ &\{\
    t_1 \colon \text{A} \, \text{B} \tr \text{C} \, \text{C},\,
    t_2 \colon \text{A} \tr \text{B},\,
    t_3 \colon \text{B} \tr \text{A}
\ \}&\ ), \\
\calP_2 &= (\ \{\ \text{A}, \text{B}\ \},& &\{\
    t_4 \colon \text{A} \, \text{B} \tr \text{A} \, \text{A},\,
    t_5 \colon \text{A} \tr \text{B}
\ \}&\ ).
\end{alignat*}

Consider the system $\calP_1$, and let $\Stage = \N^Q$, i.e.,
$\Stage$ contains all configurations. Transitions $t_2$ and $t_3$ never become
dead at $\multiset{\text{A}}$ and can thus never be included in any $U$.
Transition $t_1$ eventually becomes dead, as shown by
the linear ranking function $r(C) = C(\text{A}) + C(\text{B})$ for $U = \{t_1\}$.
But for this $U$, the condition~\ref{itm:layer2} for layer functions is not satisfied,
as $\sem{\disa{U}} \ni \multiset{\text{A}, \text{A}} \trans{t_2} \multiset{\text{A}, \text{B}} \not\in \sem{\disa{U}}$,
so $\sem{\disa{U}} \neq \sem{\dead{U}}$.
Therefore no layer function exists for this $U$.

Consider now the system $\calP_2$, again with $\Stage = \N^Q$,
and let $U = \{t_5\}$.
Once $t_5$ is disabled, there is no agent in $\text{A}$, so both $t_4$ and $t_5$ are dead.
So $\sem{\disa{U}} = \sem{\dead{U}}$.
The linear layer function $\ell(C) = C(\text{A})$ satisfies $\linlayerfun(U, \vec{a})$,
showing that $U$ eventually becomes dead.
As $C \trans{t_4 t_5} C$ for $C = \multiset{\text{A}, \text{B} }$,
there is no ranking function $r$ for this $U$, which would need to satisfy $r(C) < r(C)$.

For our implementation of $\EventDead(\Stage)$, we therefore combine both approaches.
We first compute (in polynomial time) the unique maximal set $U$ for which there is a linear ranking function.
If this $U$ is non-empty, we return it, and otherwise compute a set $U$ of maximal size for which there is a linear layer function.

\section{Experimental results} \label{sec:eval}
\newcommand{\RedStar}{{\color{red} *}}
\newcommand{\UsedOverapprox}{\hyperlink{UsedOverapproxExplained}{\RedStar}}

\begin{figure}[htbp]
    \noindent
    \begin{minipage}[t]{.52\linewidth}
        \begin{tabular}[t]{|p{3.5cm}|r|r|r|} \hline
            \rowcolor{lightergray}\multicolumn{4}{|c|}{\textbf{Population protocols} (correctness)}\\[0.05cm]
            \rowcolor{lightergray}\textbf{Parameters} & $|Q|$ & $|T|$ & \textbf{Time}\\ \hline \hline

            \multicolumn{4}{|l|}{Broadcast \cite{CDFS11,BEJM17} \UsedOverapprox} \\ \hline
            & 2 & 1 & $<1s$ \\ \hline \hline

            \multicolumn{4}{|l|}{Majority (\Cref{ex:majority})\cite{BEJM17} \UsedOverapprox} \\ \hline
            & 4 & 4 & $<1s$ \\ \hline \hline

            \multicolumn{4}{|l|}{Majority \cite[Ex.~3]{BEK18} \UsedOverapprox} \\ \hline
            & 5 & 6 & $<1s$ \\ \hline \hline

            \multicolumn{4}{|l|}{Majority \cite{AGV15} (``fast \& exact'')} \\ \hline
            $m$=13, $d$=1& 16 & 136 & $4s$ \\
            $m$=21, $d$=1 (TO: 23,1)& 24 & 300 & $466s$ \\
            $m$=21, $d$=20 (TO: 23,22)& 62 & 1953 & $3301s$ \\ \hline \hline

            \multicolumn{4}{|l|}{Flock-of-birds \cite{CMS10,BEJM17} \UsedOverapprox: $x \ge c$} \\ \hline
            $c = 20$ & 21 & 210 & $5s$ \\
            $c = 40$ & 41 & 820 & $45s$ \\
            $c = 60$ & 61 & 1830 & $341s$ \\
            $c = 80$ (TO: $c = 90$)& 81 & 3240 & $1217s$ \\ \hline \hline

            \multicolumn{4}{|l|}{Flock-of-birds \cite[Sect.~3]{BEJ18}: $x \ge c$} \\ \hline
            $c = 60$ & 8 & 18 & $15s$ \\
            $c = 90$ & 9 & 21 & $271s$ \\
            $c = 120$ (TO: $c = 127$)& 9 & 21 & $2551s$ \\ \hline \hline

            \multicolumn{4}{|l|}{Flock-of-birds \cite[\emph{threshold-n}]{CDFS11,BEJM17} \UsedOverapprox: $x \ge c$} \\ \hline
            $c = 10$ & 11 & 19 & $<1s$ \\
            $c = 15$ & 16 & 29 & $1s$ \\
            $c = 20$ (TO: $c = 25$)& 21 & 39 & $18s$ \\ \hline \hline

            \multicolumn{4}{|l|}{Threshold \cite{AADFP04}\cite[$v_{\mathrm{max}}{=}c+1$]{BEJM17} \UsedOverapprox: $\vec{a} \cdot \vec{x} \geq c$} \\ \hline
            $c=2$& 28 & 288 & $7s$ \\
            $c=4$ & 44 & 716 & $26s$ \\
            $c=6$ & 60 & 1336 & $107s$ \\
            $c=8$ (TO: $c=10$)& 76 & 2148 & $1089s$ \\ \hline \hline

            \multicolumn{4}{|l|}{Threshold \cite{BEJ18} (``succinct''): $\vec{a} \cdot \vec{x} \geq c$} \\ \hline
            $c = 7$ & 13 & 37 & $2s$ \\
            $c = 31$ & 17 & 55 & $11s$ \\
            $c = 127$ & 21 & 73 & $158s$ \\
            $c = 511$ (TO: $c = 1023$)& 25 & 91 & $2659s$ \\ \hline \hline

            \multicolumn{4}{|l|}{Remainder \cite{BEJM17} \UsedOverapprox: $\vec{a} \cdot \vec{x} \equiv_m c$} \\ \hline
            $m = 5$ & 7 & 20 & $<1s$ \\
            $m = 15$ & 17 & 135 & $34s$ \\
            $m = 20$ (TO: $m = 25$)& 22 & 230 & $1646s$ \\ \hline
        \end{tabular}
    \end{minipage}%
    \begin{minipage}[t]{.48\linewidth}
        \begin{tabular}[t]{|p{3.3cm}|r|r|r|} \hline
            \rowcolor{lightergray}\multicolumn{4}{|c|}{\textbf{Population protocols} (stable cons.)}\\[0.05cm]
            \rowcolor{lightergray}\textbf{Parameters} & $|Q|$ & $|T|$ & \textbf{Time} \\ \hline \hline

            \multicolumn{4}{|l|}{Approx. majority \cite{Cardelli2012} (Cell cycle sw.) \UsedOverapprox} \\ \hline
            & 3 & 4 & $<1s$ \\ \hline \hline

            \multicolumn{4}{|l|}{Approx. majority \cite{LLMR17} (Coin game) \UsedOverapprox} \\ \hline
            $k=3$ & 2 & 4 & $<1s$ \\ \hline \hline

            \multicolumn{4}{|l|}{Approx. majority \cite{Mor58} (Moran proc.) \UsedOverapprox} \\ \hline
            & 2 & 2 & $<1s$ \\ \hline

            \multicolumn{4}{l}{} \\[-0.19cm] \hline

            \rowcolor{lightergray}\multicolumn{4}{|c|}{\textbf{Leader election/Mutex algorithms}}\\[0.05cm]
            \rowcolor{lightergray}\textbf{Processes} & $|Q|$ & $|T|$ & \textbf{Time} \\ \hline \hline

            \multicolumn{4}{|l|}{Leader election \cite{IsraeliJ90} (Israeli-Jalfon)} \\ \hline
            20 & 40 & 80 & $7s$\\
            60 & 120 & 240 & $1493s$\\
            70 (TO: $80$) & 140 & 280 & $3295s$\\ \hline \hline

            \multicolumn{4}{|l|}{Leader election \cite{LeaderHerman} (Herman)} \\ \hline
            21 & 42 & 42 & $9s$\\
            51 & 102 & 102 & $300s$\\
            81 (TO: $91$) & 162 & 162 & $2800s$\\ \hline \hline

            \multicolumn{4}{|l|}{Mutex \cite{MutexArray} (Array)} \\ \hline
            2 & 15 & 95 & $2s$\\
            5 & 33 & 239 & $5s$\\
            10 (TO: $11$) & 63 & 479 & $938s$\\ \hline \hline

            \multicolumn{4}{|l|}{Mutex \cite{MutexBurns} (Burns)} \\ \hline
            2 & 11 & 75 & $1s$\\
            4 & 19 & 199 & $119s$\\
            5 (TO: $6$) & 23 & 279 & $2232s$\\ \hline \hline

            \multicolumn{4}{|l|}{Mutex \cite{MutexDijkstra} (Dijkstra)} \\ \hline
            2 & 19 & 196 & $66s$ \\
            3 (TO: $4$)& 27 & 488 & $3468s$ \\ \hline \hline

            \multicolumn{4}{|l|}{Mutex \cite{MutexLehmannRabin} (Lehmann Rabin)} \\ \hline
            2 & 19 & 135 & $3s$\\
            5 & 43 & 339 & $115s$\\
            9 (TO: $10$) & 75 & 611 & $2470s$\\ \hline \hline

            \multicolumn{4}{|l|}{Mutex \cite{Peterson81} (Peterson)} \\ \hline
            2 & 13 & 86 & $2s$ \\ \hline \hline

            \multicolumn{4}{|l|}{Mutex \cite{MutexSzymanski} (Szymanski)} \\ \hline
            2 & 17 & 211 & $10s$ \\
            3 (TO: $4$)& 24 & 895 & $667s$ \\ \hline
        \end{tabular}
    \end{minipage}
    \caption{Columns $|Q|$, $|T|$, and \textbf{Time} give the number of states and non-silent transitions, and the time for verification. Population protocols are verified for an infinite set of configurations. For parametric families, the smallest instance that could not be verified within one hour is shown in brackets, e.g. (TO: $c = 90$). Leader election and mutex algorithms are verified for one configuration. The number of processes leading to a timeout is given in brackets, e.g. (TO: 10).}
    \label{fig:eval}
\end{figure}

We implemented the procedure of \Cref{sec:stagegraph_construction} on top of the SMT solver \emph{Z3}~\cite{z3}, and use the Owl~\cite{Owl} and HOA~\cite{HOA} libraries for translating LTL formulas.
The resulting tool automatically constructs stage graphs that verify stable termination properties for replicated systems. We evaluated it on two sets of benchmarks, described below. The first set contains population protocols, and the second leader election and mutual exclusion algorithms. All tests where performed on a machine with an Intel Xeon CPU E5-2630 v4 @ 2.20GHz and 8GB of RAM.
The results are depicted in \Cref{fig:eval} and can be reproduced by the certified
artifact~\cite{Artifact}.
For parametric families of replicated systems, we always report the
largest instance that we were able to verify with a timeout of one hour.
\hypertarget{UsedOverapproxExplained}{For $\IndOverapprox$, from the approaches in \Cref{sec:setconfdead}, we use $\IndOverapprox^0$ in the examples marked with \RedStar{} and $\IndOverapprox^\infty$ otherwise.}
Almost all constructed stage graphs are a chain with at most 3 stages. The only exceptions are the stage graphs for the approximate majority protocols that contained a binary split and 5 stages.
The size of the Presburger formulas increases with increasing size of the replicated system. In the worst case, this growth can be exponential.
However, the growth is linear in all examples marked with \RedStar{}.

\paragraph{Population protocols.} Population protocols~\cite{AADFP04,AADFP06} are replicated systems that compute Presburger predicates following the computation-as-consensus paradigm~\cite{AAER07}. Depending on whether the initial configuration of agents satisfies the predicate or not, the agents of a correct protocol eventually agree on the output ``yes'' or ``no'', almost surely. \Cref{ex:majority} can be interpreted as a population protocol for the majority predicate $\AY > \AN$, and the two stable termination properties that verify its correctness are described in \Cref{ex:majority-property}. To show that a population protocol correctly computes a given predicate, we thus construct two Presburger stage graphs for the two corresponding stable termination properties.  In all these examples, correctness is proved for an infinite set of initial configurations.

Our set of benchmarks contains a broadcast protocol \cite{CDFS11}, three majority protocols (\Cref{ex:majority}, \cite[Ex.~3]{BEK18}, \cite{AGV15}), and multiple instances of parameterized families of protocols, where each protocol computes a different instance of a parameterized family of predicates\footnote{Notice that for each protocol we check correctness for all inputs; we cannot yet automatically verify that infinitely many protocols are correct, each of them for all possible inputs.}. These include various \emph{flock-of-birds} protocol families (\cite{CMS10}, \cite[Sect.~3]{BEJ18}, \cite[\emph{threshold-n}]{CDFS11}) for the family of predicates $x \geq c$ for some constant $c \geq 0$; two families for threshold predicates of the form $\vec{a} \cdot \vec{x} \geq c$~\cite{AADFP04,BEJ18}; and one family for remainder protocols of the form $\vec{a} \cdot \vec{x} \equiv_m c$~\cite{BEJM17}. Further, we check approximate majority protocols (\cite{Cardelli2012,Mor58}, \cite[\emph{coin game}]{LLMR17}). As these protocols only compute the predicate with large probability but not almost surely, we only verify that they always converge to a stable consensus.

\paragraph{Comparison with \cite{BEJM17}.}
The approach of \cite{BEJM17} can only be applied to so-called \emph{strongly-silent} protocols. However, this class does not contain many fast and succinct protocols recently developed for different tasks~\cite{AlistarhG18,BEGHJ20,BEJ18}.

We are able to verify all six protocols reported in \cite{BEJM17}. Further, we are also able to verify the fast Majority \cite{AGV15} protocol as well as the succinct protocols Flock-of-birds \cite[Sect.~3]{BEJ18} and Threshold \cite{BEJ18}. All three protocols are not strongly-silent. Although our approach is more general and complete, the time to verify many strongly-silent protocol does not differ significantly between the two approaches. Exceptions are the Flock-of-birds \cite{CMS10} protocols where we are faster (\cite{BEJM17} reaches the timeout at $c=55$) as well as the Remainder and the Flock-of-birds-threshold-$n$ protocols where we are substantially slower (\cite{BEJM17} reaches the timeout at $m=80$ and $c=350$, respectively). Loosely speaking, the approach of \cite{BEJM17} can be faster because they compute inductive overapproximations using an iterative procedure instead of $\PReach$. In some instances already a very weak overapproximation, much less precise than $\PReach$, suffices to verify the result. Our procedure can be adapted to accommodate this (it essentially amounts to first running the procedure of \cite{BEJM17}, and if it is inconclusive then run ours).

\paragraph{Other distributed algorithms.}
We have also used our approach to verify arbitrary LTL liveness properties of non-parameterized systems with arbitrary communication structure. For this we apply
standard automata-theoretic techniques and construct a product of the system and a \emph{limit-deterministic Büchi automaton} for the negation of the property. Checking that no fair runs of the product are accepted by the automaton reduces to checking a stable termination property.

Since we only check correctness of one single finite-state system, we can also apply a probabilistic model checker based on state-space exploration. However, our technique delivers a stage graph,
which plays two roles. First, it gives an explanation of why the property holds in terms of invariants and ranking functions, and second, it is a certificate of correctness that can be efficiently checked by independent means.

We verify liveness properties for several leader election and mutex algorithms from the literature~\cite{MutexDijkstra,MutexArray,LeaderHerman,IsraeliJ90,MutexLehmannRabin,MutexBurns,Peterson81,MutexSzymanski} under the assumption of a probabilistic scheduler. For the leader election algorithms, we check that a leader is eventually chosen; for the mutex algorithms, we check that the first process enters its critical section infinitely often.

\paragraph{Comparison with PRISM~\cite{KNP11}.}
We compared execution times for verification by our technique and by PRISM on the same models. While PRISM only needs a few seconds to verify instances of the mutex algorithms \cite{MutexDijkstra,MutexArray,MutexLehmannRabin,MutexBurns,Peterson81,MutexSzymanski} where we reach the time limit, it reaches the memory limit for the two leader election algorithms~\cite{LeaderHerman,IsraeliJ90} already for 70 and 71 processes, which we can still verify.

\section{Conclusion and further work} \label{sec:conclusion}
We have presented stage graphs, a sound and complete technique for the verification of stable termination properties of replicated systems, an important class of parameterized systems. Using deep results of the theory of Petri nets, we have shown that Presburger stage graphs, a class of stage graphs whose correctness can be reduced to the satisfiability problem of Presburger arithmetic, are also sound and complete. This provides a decision procedure for the verification of termination properties, which is of theoretical nature since it involves a blind enumeration of candidates for Presburger stage graphs. For this reason, we have presented a technique for the algorithmic construction of Presburger stage graphs, designed to exploit the strengths of SMT-solvers for existential Presburger formulas, i.e., integer linear constraints. Loosely speaking,  the technique searches for \emph{linear} functions certifying the progress between stages, even though only the much larger class of Presburger functions guarantees completeness.

We have conducted extensive experiments on a large set of benchmarks. In particular, our approach is able to prove correctness of nearly all the standard protocols described in the literature, including several protocols that could not be proved by the technique of~\cite{BEJM17}, which only worked for so-called strongly-silent protocols. We have also successfully applied the technique to some self-stabilization algorithms, leader election and mutual exclusion algorithms.

Our technique is based on the mechanized search for invariants and ranking functions. It avoids the use of state-space exploration as much as possible. For this reason, it also makes sense as a technique for the verification of liveness properties of non-parameterized systems with a finite but very large state space.

%
%
%
\bibliographystyle{splncs04}
\bibliography{bibliography}

\clearpage
\appendix
\section{Appendix} \label{sec:appendix}
\subsection[Missing proofs for Section 2]{Missing proofs for \Cref{sec:preliminaries}}\label{app:preliminaries}

We show that the qualitative model checking problem is not semi-decidable.
The result holds even for the subclass of replicated systems of arity 2 (i.e., for population protocols)
and when $\Init = \sem{\varphi_1}$ and the LTL formula is of the form $\varphi = \Box \varphi_2 \lor \Diamond \varphi_3$,
where $\varphi_1,\varphi_2$ and $\varphi_3$ are quantifier-free Presburger predicates with atomic formulas of the form $q{=}0$, $q{=}1$, or $q{\ge}1$ for $q \in Q$.

\qualundecidable*
\begin{proof}
A two-counter \emph{Minsky machine} $\mathcal{M}$ is a finite sequence of labeled instructions
\[
   \ell_1 : \mathit{ins}_1, \quad \ldots \quad ,\,\ell_m : \mathit{ins}_m,\, \ell_{m+1} : \textbf{halt}
\]
where every $\mathit{ins}_i$ is either a \emph{Type~I} instruction of the form
\[
    \textbf{inc}\ c_j;\ \textbf{goto}\ \ell_k
\]
where $j \in \{1,2\}$ and $1\leq k \leq m+1$, or a \emph{Type~II} instruction of the form
\[
\textbf{if}\ c_j{=}0 \ \textbf{then\ goto}\ \ell_k \ \textbf{else\ dec}\ c_j;\ \textbf{goto}\ \ell_n
\]
where $j \in \{1,2\}$ and $1\leq k,n \leq m+1$.

A computation of $\mathcal{M}$ starts by executing the first instruction with both counters $c_1$ and $c_2$ initialized to zero. The problem of determining whether $\mathcal{M}$ \emph{halts}, i.e., eventually executes the $\textbf{halt}$ instruction, is undecidable~\cite{Minsky:book}. Consequently, the problem of whether $\mathcal{M}$ does \emph{not} halt is not even semi-decidable.

We prove our theorem by reducing the non-halting problem for two-counter Minsky machines to the qualitative model checking problem. For a given Minsky machine~$\mathcal{M}$, let $L_{I}$ and $L_{II}$ be the sets of all indices $i \in \{1,\ldots,m\}$ such that $\mathit{ins}_i$ is a Type~I and Type~II instruction, respectively.
We construct a replicated system $\calP = (Q,T)$ where
\[
 Q \ \defeq\ \{q_1,\ldots,q_{m+1},Z_1,O_1,Z_2,O_2\} \cup \{\hat{q}_i \mid i \in L_{II}\},
\]
and $T$ is the (least) set of transitions satisfying the following:
\begin{itemize}
	\item For every Type~I instruction of the form
	\[
        \ell_i:\ \textbf{inc}\ c_j;\ \textbf{goto}\ \ell_k
	\]
	there is a transition $q_i \, Z_j \tr q_k \, O_j$.
	\item For every Type~II instruction of the form
	\[
        \ell_i:\ \textbf{if}\ c_j{=}0 \ \textbf{then\ goto}\ \ell_k \ \textbf{else\ dec}\ c_j;\ \textbf{goto}\ \ell_n
	\]
	there are transitions $q_i \, Z_j \tr \hat{q}_i \, Z_j$, $\hat{q}_i \, Z_j \tr q_k \, Z_j$, and $q_i \, O_j \tr q_n \, Z_j$.
\end{itemize}
Consider the following Presburger formulas:
\begin{align*}
    \mathit{Init}  \ &\defeq \ q_1 {=} 1 \ \land\ Z_1{\ge}1 \ \land\ Z_2{\ge}1 \ \land\ \bigwedge_{q \in Q \setminus \{q_1,Z_1,Z_2 \}} q {=} 0 \\
	\mathit{Overflow} \  &\defeq \ \bigvee_{i \in L_I} (q_i {=} 1 \wedge Z_{j_i}{=} 0)\\
	\mathit{Cheat}  \ &\defeq \ \bigvee_{i \in L_{II}} (\hat{q}_i {=} 1 \wedge O_{j_i} {\geq} 1)
\end{align*}
In the above formulas, we use $j_i \in \{1,2\}$ to denote the counter used by instruction $\mathit{ins}_i$.
Furthermore, let $\Init \defeq \sem{\mathit{Init}}$ and
\[
    \varphi \ \defeq\  \Box (q_{m+1} {=}0) \ \lor\  \Diamond (\mathit{Overflow} \lor \mathit{Cheat}).
\]
We claim that $\mathcal{M}$ does not halt iff $\Pr[C, \varphi] = 1$ for every configuration $C \in \Init$.

Suppose $\mathcal{M}$ does not halt. Let $C \in \Init$. As $C$ satisfies the formula $\mathit{Init}$, it has precisely one agent in state $q_1$, at least one agent in each state $Z_1$ and $Z_2$, and no agents elsewhere. The transitions of $\calP$ are constructed so that they allow for simulating $\mathcal{M}$ from~$C$. In every configuration $C'$ reachable from $C$, there is precisely one agent in a state of $\{q_1,\ldots,q_{m+1}\} \cup \{\hat{q}_i \mid i \in L_{II}\}$, and the values of $c_1$ and $c_2$ are represented by $C'(O_1)$ and $C'(O_2)$, respectively. Since $C'(O_1)$ and $C'(O_2)$ are bounded, the simulation may fail due to a \emph{counter overflow} when some Type~I instruction tries to increase a counter $c_j$ (i.e., rewrite $Z_j$ into $O_j$) but no agent in $Z_j$ is available. This is captured by the formula $\mathit{Overflow}$. Furthermore, the simulation of a Type~II instruction is not necessarily faithful, because the transition $q_i \, Z_j \tr \hat{q}_i \, Z_j$ can be executed even if there is an agent in state $O_j$ in the current configuration (i.e., the corresponding counter value is positive). This is detected by the formula $\mathit{Cheat}$. Hence, if $\mathcal{M}$ does not halt, then every run initiated in $C$ either does \emph{not} correspond to a faithful simulation of $\mathcal{M}$, i.e., visits a configuration satisfying $\mathit{Overflow}$ or $\mathit{Cheat}$, or simulates $\mathcal{M}$ faithfully, i.e., the state $q_{m+1}$ does not occur in any configuration visited by the run. Hence, all runs initiated in $C$ satisfy the formula $\varphi$.

If $\mathcal{M}$ halts, then the instruction $\textbf{halt}$ is executed after a \emph{finite} computation along which the counters are increased only to some finite values. Hence, for all sufficiently large $n$, there exist a configuration $C \in \Init$ with $C(Z_1) = C(Z_2) = n$ and a finite path initiated in $C$ corresponding to a faithful simulation of $\mathcal{M}$. Note that the last configuration $C'$ of this path (where $C'(q_{m+1}) = 1$) has only one successor $C'$, i.e., the self-loop $C' \trans{} C'$ is inevitably selected with probability~$1$. The probability of executing this path (and the associated run) is positive, and the run does not satisfy the formula $\varphi$. Hence, $\Pr[C,\varphi] <  1$.
\end{proof}

\subsection[Missing proofs for Section 3]{Missing proofs for \Cref{sec:stagegraphs}}\label{app:stagegraphs}

\propCert*

\begin{proof}
\noindent ($\Rightarrow$): Assume $\C$ leads to $\C'$. By definition of the ``leads to'' relation, for every $C \in \C$, there exists $C' \in \C'$ such that $C \trans{*} C'$. Hence, the function defined by $f(C) = 1$ if $C \in \C \setminus\C'$, and $f(C)=0$ otherwise is a certificate for $\C \leadsto \C'$. \medskip

\noindent ($\Leftarrow$): Assume there is a certificate for $\C \leadsto \C'$. We claim that for every $C \in \C$, there exists $C' \in \C'$ such that $C \trans{*} C'$.  Assume the contrary. By assumption, the definition of certificates and as $\C$ is inductive, there are configurations $C_0, C_1, \ldots \in \C \setminus \C'$ such that $C = C_0 \trans{*} C_1 \trans{*} \cdots$ and $f(C_i) > f(C_{i+1})$ for every $i \geq 0$. This is impossible as the codomain of $f$ is $\N$, which proves the claim.

We may now prove that $\C$ leads to $\C'$. Let $\rho$ be a fair run
starting at some configuration of $\C$. Since $\C$ is inductive,
$\rho$ only visits configurations of $\C$. Further, since all
configurations visited by $\rho$ have the same size, some
configuration $C \in \C$ is visited infinitely often.  By the claim,
there exists a sequence $C \trans{\sigma} C'$ such that
$C' \in \C'$. We show that $\rho$ visits $C'$, by induction on
$|\sigma|$. If $|\sigma| = 0$, then $C' = C$ and we are done. Assume
$\sigma = t \tau$ and $C \trans{t} D \trans{\tau} C'$, where $t \in
T$. By fairness, $D$ occurs infinitely often in $\rho$. Since
$|\tau| = |\sigma|-1$, we can apply the induction hypothesis to $\tau$
and conclude that $C'$ occurs infinitely often in $\rho$. \qed
\end{proof}

\propComp*

\begin{proof}
$(\Leftarrow$): Assume $\calP$ has a stage graph for $\Pi$.
Let $\Bottom$ be a terminal stage of the stage graph. By condition~\ref{st4}, $\Bottom \models \phipost^i$ for some $i$, and by inductiveness $\Bottom \models \Box \phipost^i$
Let $\Leaves$ be the union of the terminal stages of the stage graph. We have $\Leaves \models \bigvee_{i=1}^k \Box \phipost^i$.
By Proposition~\ref{prop:cert} and conditions~\ref{st1} and~\ref{st3} of the definition of a stage graph, every stage $\C$ satisfies $\C \leadsto \Leaves$.
Therefore every stage $\C$ satisfies $\Diamond \bigvee_{i=1}^k \Box \phipost^i$.
By condition~\ref{st2}, we have that $\sem{\phipre} \models \Diamond \bigvee_{i=1}^k \Box \phipost^i$. Thus
$C \models \Diamond \bigvee_{i=1}^k \Box \phipost^i$ for any configuration $C \in \sem{\phipre}$, and hence $\calP \models \varphi_\Pi$.\medskip

\noindent $(\Rightarrow$): Assume $\calP$ satisfies $\Pi$. Consider a stage graph with $k+1$ stages: an initial stage $\C_{in}$ containing the set of all configurations reachable from $\sem{\phipre}$,
and a terminal stage $\C_{f_i}$, for every $1 \le i \le k$, containing all configurations satisfying $\Box \phipost^i$. Conditions~\ref{st1}, \ref{st2}, and~\ref{st4} hold by definition.
Since $\calP$ satisfies $\Pi$, every fair run from a configuration of the initial stage $\C_{in}$ eventually visits a terminal stage $\C_{f_i}$,
and therefore $\C_{in}$ leads to $(\C_{f_1} \cup \ldots \cup \C_{f_k})$. Consequently, \Cref{prop:cert} yields a certificate for $\C_{in} \leadsto (\C_{f_1} \cup \ldots \cup \C_{f_k})$. \qed
\end{proof}

\thmStageGraph*

\begin{proof}
  \newcommand{\I}{\mathcal{I}}%
  \newcommand{\B}{\mathcal{B}}%
  \renewcommand{\R}{\mathcal{R}}%

  We say that a configuration $C$ is \emph{bottom} if $C \trans{*} D$
  implies $D \trans{*} C$. Let $\B$ be the set of all bottom
  configurations of $\PP$. Let $\xleftrightarrow{*}$ denote the
  \emph{mutual reachability relation} defined by $C
  \xleftrightarrow{*} D \defiff (C \trans{*} D \land D \trans{*}
  C)$. It is known from~\cite[Thm.~13 and Prop.~14]{EGLM17} that both
  $\xleftrightarrow{*}$ and $\B$ are (effectively) \Epres. Let $\S_i
  \defeq \B \cap \sem{\Box \phipost^i}$ for every $i \in [n]$. Note
  that $\S_i$ is \Epres{} since it can be written as
  $$ \S_i = \left\{C \in \B : \forall D\ [(C \xleftrightarrow{*} D)
    \implies (D \models \phipost^i)]\right\}.
  $$ Let $\S \defeq \S_1 \cup \cdots \cup \S_n$ and let $\I \defeq
  \sem{\phipre}$. Note that $\S$ and $\I$ are \Epres{}. Since $\PP$
  satisfies $\Pi$, we have $\postC^*(\I) \cap (\B \setminus \S) =
  \emptyset$. Therefore, by~\cite[Lem.~9.1]{Ler12}, there exists an
  inductive \Epres{} set $\I' \supseteq \I$ such that $\postC^*(\I')
  \cap (\B \setminus \S) = \emptyset$. Since any set of configurations
  leads to $\B$, this implies $\I' \leadsto \S$.

  The directed acyclic graph made of $\I'$ with $\S_1, \ldots, \S_n$
  as its successors is a stage graph for $\Pi$. Indeed:
  \begin{enumerate}
  \item $\I'$ and $\S_i$ are inductive, where the latter follows by
    definition;

  \item $\sem{\phipre} = \I$ is a subset of stage $\I'$;

  \item $\I' \leadsto \S = (\S_1 \cup \dots \cup \S_n)$ holds, by the above;
    and

  \item $\S_i \models \phipost^i$ by definition of $\S_i$.
  \end{enumerate}

  It remains to exhibit a \Epres{} certificate. Since $\I'$ and $\S$
  are both \Epres, \cite[Cor.~XI.3]{Ler13} yields a bounded language
  $L = w_1^* w_2^* \cdots w_k^* \subseteq T^*$ such that $\I'
  \subseteq \preC_L(\S)$. Let $\preC_L(\S)$ be the set of configurations
  $C$ such that $C \trans{w} C'$ for some $C' \in \S$ and $w \in L$, and $L'$ be the language made of all
  sequences of $L$ and their suffixes. We have $\I' \subseteq
  \preC_L(\S) \subseteq \preC_{L'}(\S)$. For every $C \in \I'$, let
  $f(C) \defeq |\sigma_C|$ where $\sigma_C \in L'$ is a shortest
  sequence such that:
  $$C \trans{\sigma_C} D \text{ for some } D \in \S.$$ Since $\I'$ is
  inductive and since $L'$ is closed under suffixes, if $\sigma_C = t
  \tau$ and $C \trans{t} C'$, then we have $f(C) = f(C') + 1$. Hence,
  $(f, 1)$ is a bounded certificate for $\I' \leadsto \S$.

  It remains to construct a \Epres{} formula $\varphi(C, \ell)$ that
  holds if{}f $\ell = f(C)$. We only consider the case where $L = w^*$
  for some finite sequence $w$; the generalization to $L = w_1^* w_2^*
  \cdots w_k^*$ being straightforward.

  A simple induction on the length of $w$ shows that the set of configurations that
  enable $w$ has a unique minimal configuration $C_w$. Further, also by induction
  on $w$ there exists a vector $\effect{w}
  \in \Z^{Q}$ such that $C \trans{w} C + \effect{w}$ for every $C \geq
  C_w$. More precisely, $\effect{w} = \sum_{i=1}^{|w|}
  \effect{w_i}$. It follows that a configuration $C$ enables sequence
  $w^k$ if{}f $C + i \cdot \effect{w} \geq C_w$ for every $0 \leq i <
  k$. Furthermore, if $C$ enables $w^k$, then we have: $$C \trans{w^k}
  C + k \cdot \effect{w}.$$ This condition is expressible by a
  Presburger formula $\enab_w(C, k)$. Let $\textit{suff}(w)$ denote
  the set of sufffixes of $w$, and let
  \begin{alignat*}{3}
    \textit{FS}(C, \ell)
    &\defeq
    &\bigvee_{w' \in \textit{suff}(w)} \enab_{w'}(C, 1) \wedge \exists k
    \colon & \enab_{w}(C + \effect{w'}, k)\ & \wedge \\[-12pt]
    &&& C + \effect{w'} + k \cdot \effect{w} \in \S\ & \wedge \\[0pt]
    &&& \ell = |w'| + k.
  \end{alignat*}
  Formula $\textit{FS}(C, \ell)$ holds if{}f there exists a sequence
  $\sigma \in L'$ of length $\ell$ such that $C \trans{\sigma} C'$ and
  $C' \in \S$. Let
  $$
  \varphi(C, \ell) \defeq \textit{FS}(C, \ell) \wedge \forall \ell'
  \colon \textit{FS}(C, \ell') \rightarrow (\ell' \geq \ell).
  $$ Formula $\varphi(C, \ell)$ holds if{}f $\ell$ is the length of a
  shortest sequence $\sigma \in L'$ such that $C \trans{\sigma}C'$ and
  $C' \in \S$, as desired. \qed
\end{proof}

\thmStageGraphDec*

\begin{proof}
  First we observe that for any two configurations $C,C'$, we have $C \trans{} C'$ if{}f there exists a transition $t$
  such that $C \ge \pre{t}$ and $C' = C + \effect{t}$.
  With that, checking the inductiveness of a \Epres{} stage $\Stage$ reduces to checking satisfiability of this sentence:
  $$
  \forall C \colon C \in \Stage \rightarrow \bigwedge_{t \in T} \left( C \ge \pre{t} \rightarrow \left( C + \effect{t} \right) \in \Stage \right).
  $$
  Checking whether $\sem{\phipre}$ is included in the union of all stages reduces to checking satisfiability of this sentence:
  $$
  \forall C \colon \phipre(C) \rightarrow \bigvee_{\mathclap{\text{stage}~\Stage}} C \in \Stage.
  $$
  Let $\Stage$ be a terminal stage of the graph. Checking that $\Stage \models \phipost^i$ for some $i$ reduces to checking
  satisfiablity of this sentence:
  $$
  \bigvee_{i=1}^k \forall C: C \in \Stage \rightarrow \phipost^i(C).
  $$
  Let $(f, k)$ be a Presburger certificate and $\varphi(\mathbf{x}, y)$ be the existential Presburger formula given for $f$.
  Checking that $\varphi$ actually describes some function $f$ reduces to checking satisfiability of this sentence:
  $$
  \left( \forall C \colon \exists y \colon \varphi(C,y) \right) \land
  \left( \forall C, y, y' \colon \left( \varphi(C,y) \land \varphi(C,y') \right) \rightarrow y = y' \right).
  $$
  Since we have the silent transition $(\multiset{}, \multiset{}) \in T$ that is always enabled, it suffices to check $(f,k)$ for sequences of length exactly $k$ instead of at most $k$.
  We now obtain that $(f, k)$ is a Presburger certificate for $\Stage \leadsto (\Stage_1 \cup \ldots \cup \Stage_n)$
   if{}f the following sentence is satisfiable:
  \begin{align*}
  \forall C_0, y \colon \exists C_1,\ldots,C_k,y' \colon
  \left[\varphi(C,y) \land C \in \Stage \land \neg \bigvee_{i=1}^n C \in \Stage_i \right] \rightarrow \\
  \left[
      \varphi(C_k,y') \land y > y' \land
      \bigwedge_{i=0}^{k-1} \bigvee_{t \in T} \left( C_i \ge \pre{t} \land C_{i+1} = C_i + \effect{t} \right)
  \right].
  \end{align*}

  Determining if the given graph is a Presburger stage for $\Pi$ now amounts to
  determining satisfiability of the conjunction of all the constructed Presburger sentences.
  We note that if all Presburger formulas for stages and
  certificates are quantifier-free and the bounds $k$ on the certificates are
  given in unary, then the constructed Presburger sentences are of polynomial
  size and in the $\forall\exists$ fragment. \qed
\end{proof}

\subsection[Missing proofs for Section 4]{Missing proofs for \Cref{sec:stagegraph_construction}}\label{app:stagegraph_construction}

\eventdeadhardness*
\begin{proof}
  \newcommand{\init}[1]{{#1}_{\text{init}}}
\newcommand{\target}[1]{{#1}_{\text{tgt}}}
\newcommand{\clean}[1]{{#1}_{\text{clean}}}
\newcommand{\free}[1]{{#1}_{\text{free}}}
\newcommand{\out}[1]{{#1}_{\text{out}}}

Let us first establish decidability. It suffices to show decidability
of $\S \not\models \Diamond \dead{U}$, \ie, whether $C \not\models
\Diamond \dead{U}$ for some $C \in \S$. Observe that
\begin{align*}
  C \not\models \Diamond \dead{U}
  &\iff C \not\models \Diamond \bigwedge_{t \in U} \dead{t} \\
  &\iff C \models \Box \bigvee_{t \in U} \neg\dead{t}.
\end{align*}
In other words, $C \not\models \Diamond \dead{U}$ holds iff at every
configuration $C'$ reachable from $C$, some transition from $U$ is
enabled at $C'$. It has been observed in~\cite[Sect.~4]{JP19} that
this notion of liveness is equivalent to liveness of a \emph{single}
transition. Indeed, it suffices to introduce a new transition
$t^\dagger$ and two new states $p^\dagger, q^\dagger$ such that:
\begin{itemize}
\item a single agent is initially in state $p^\dagger$, while none is
  in state $q^\dagger$;

\item each transition from $U$ is altered so that it takes an agent in
  state $p^\dagger$ and moves it to state $q^\dagger$;

\item $t^\dagger$ moves an agent in state $q^\dagger$ to state
  $p^\dagger$.
\end{itemize}
This way, $C \models \Box \bigvee_{t \in U} \neg\dead{t}$ holds in the
original system iff $C \models \Box \neg\dead{t^\dagger}$ holds in the
altered system. Moreover, the alteration of $\S$ remains semilinear as
one can simply add the conjunct $(p^\dagger = 1 \land q^\dagger = 0)$
to $\phi$.

By~\cite[Thm.~2]{JP19}, the problem of determining whether $\C \models
\Box \neg\dead{t^\dagger}$ holds for some $\C \in \S$, is
decidable. Although the statement of \cite[Thm.~2]{JP19} only covers
the specific case of $\S = \N^Q$, its proof explicitly handles any
effective semilinear set $\S$. Therefore, this establishes
decidability of our problem.\bigskip

Let us now show \PSPACE-hardness. Observe that replicated systems are
Petri nets where states correspond to places and where transitions
correspond to transitions and arcs. In fact, replicated systems amount
\emph{precisely} to the class of \emph{1-conservative} Petri nets,
\ie, where every transition produces as many tokens as it
consumes. Since the reachability problem for 1-conservative Petri nets
is \PSPACE-complete~\cite{JLL77}, the same holds for the reachability
problem for replicated systems defined as:

\begin{quote}
  \begin{tabular}{lp{9.5cm}}
    \textsc{Input:} &
    a replicated system $\PP = (Q, T)$ and configurations $C, C' \in
    \N^Q$, \\[2pt]
    
    \textsc{Output:} &
    does $C \trans{*} C'$ hold?
  \end{tabular}
\end{quote}
We give a (many-one) reduction from this problem to the following
variant of the \emph{partial structural liveness problem} for
replicated systems:

\begin{quote}
  \begin{tabular}{lp{9.5cm}}
    \textsc{Input:} &
    a replicated system $\PP = (Q, T)$ and a transition $t \in T$, \\[2pt]

    \textsc{Output:} &
    does $\N^Q \not\models \Diamond \dead{t}$ hold?
  \end{tabular}
\end{quote}

Let us fix a replicated system $\PP = (Q, T)$ and configurations
$\init{C}, \target{C} \in \N^Q$. We design a replicated system $\PP' =
(Q', T')$ and a transition $\target{t} \in T$ such that:
$$\init{C} \trans{*} \target{C} \text{ in } \PP \iff \N^{Q'}
\not\models \Diamond \dead{\target{t}} \text{ in } \PP'.$$ The
validity of this equivalence proves the proposition as it is a special
case of the problem we wish to show \PSPACE-hard. Note that we are
implicitly using the fact that $\PSPACE = \NPSPACE$ as we deal with
``$\not\models$'' instead of ``$\models$''.

\paragraph*{Construction.}

Let us describe $\PP'$. Its set of states is defined as $$Q' \defeq Q
\cup \{\free{q}, \out{q}\},$$ where $\free{q}$ indicates that a
``retired'' agent is ``free'' to move to a state of $Q$, and where
$\out{q}$ indicates that a ``retired'' agent is permanently retired.

Let $k \defeq |\init{C}|$. Let $\init{t} \defeq \multiset{k \cdot
  \free{q}} \mapsto \init{C}$ and let $\target{t} \defeq \target{C}
\mapsto \target{C}$. The purpose of these transitions is respectively
to generate the initial configuration $\init{C}$, and to check whether
the target configuration $\target{C}$ is present. Let $$\clean{t}
\defeq \multiset{k \cdot \free{q}, \free{q}} \mapsto \multiset{k \cdot
  \free{q}, \out{q}}.$$ The purpose of transition $\clean{t}$ is to
permanently retire agents until there are at most $k$ remaining. Let
$S \defeq \{s_q : q \in Q\}$ where $s_q \defeq \multiset{q} \mapsto
\multiset{\free{q}}$. Transitions $S$ make the system ``lossy'' in the
sense that agents can non deterministically retire from $Q$, either
temporarily or eventually permanently. Overall, the set of transitions
of $\PP'$ is defined as $$T' \defeq T \cup S \cup \{\clean{t},
\init{t}, \target{t}\}.$$

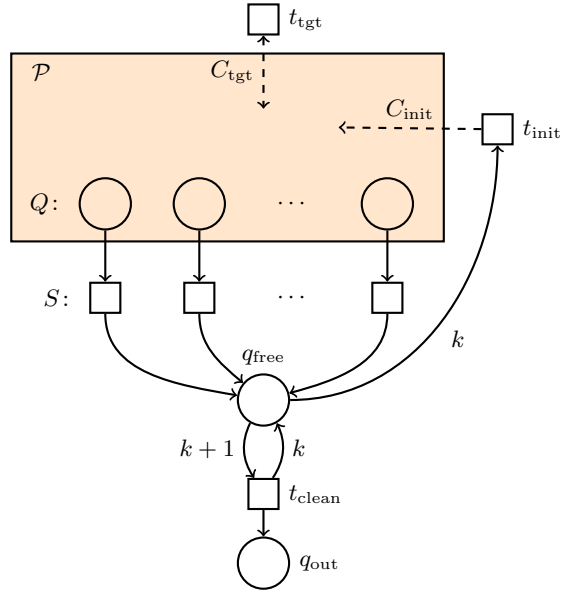
\begin{figure}[!h]
  \centering
  \begin{tikzpicture}[auto, node distance=1.25cm, thick, transform shape]
    \draw[fill=orange!20] (-1.25, -0.5) rectangle (4.5cm, 2cm);

    \node[place]              (q0) {};
    \node[place, right of=q0] (q1) {};
    \node[       right of=q1] (qi) {$\cdots$};
    \node[place, right of=qi] (qn) {};

    \node[left=1pt of q0] {$Q \colon$};

    \node[transition, below of=q0] (s0) {};
    \node[transition, below of=q1] (s1) {};
    \node[            below of=qi] (si) {$\cdots$};
    \node[transition, below of=qn] (sn) {};

    \node[left=1pt of s0] {$S \colon$};

    \path[->]
    (q0) edge node {} (s0)
    (q1) edge node {} (s1)
    (qn) edge node {} (sn)
    ;

    \node[place, below right=60pt and 10pt of q1, label=$\free{q}$] (qfree) {};
    \node[transition, below of=qfree,  label=right:$\clean{t}$]    (qclean) {};
    \node[place,      below=10pt of qclean, label=right:$\out{q}$] (qout)   {};

    \path[->]
    (s0) edge[out=-90, in=170] node {} (qfree)
    (s1) edge[out=-90, in=140] node {} (qfree)
    (sn) edge[out=-90, in=15]  node {} (qfree)
    
    (qfree)  edge[bend right] node[swap, yshift=8pt] {$k + 1$} (qclean)
    (qclean) edge[bend right] node[swap, yshift=8pt] {$k$}     (qfree)

    (qclean) edge node {} (qout)
    ;

    \node[above=3.15cm of qfree] (PP) {};
    \node at (-0.85, 1.75) {$\PP$};
    
    \node[transition, above right=15pt and of qn,
      label=right:$\init{t}$] (tinit) {};
    
    \path[->]
    (qfree) edge[out=0, in=-90] node[swap] {$k$} (tinit)
    ;

    \path[->, dashed]
    (tinit) edge node[swap] {$\init{C}$} ($(PP)+(1,0)$)
    ;

    \node[transition, above=1.1cm of PP,
      label=right:$\target{t}$] (ttarget) {};
    
    \path[<->, dashed]
    (ttarget) edge node[swap] {$\target{C}$} ($(PP)+(0,0.25)$)
    ;
  \end{tikzpicture}
  \caption{Replicated system $\PP'$ depicted as a (1-conservative) Petri net.}
  \label{fig:pspace:pn}
\end{figure}

\paragraph{General idea.}

Let us explain the idea behind the construction of $\PP'$, which is
illustrated in \Cref{fig:pspace:pn}. If $\target{C}$ is reachable from
$\init{C}$ in $\PP$, then this is also the case in $\PP'$, as it can
simulate the former. Moreover, if $\PP'$ gets stuck by choosing the
wrong transitions, then this does not yield a dead configuration, as
lossy transitions $S$ can temporarily retire agents so that $\init{t}$
resets the system to $\init{C}$. This way, transition $\target{t}$ can
occur infinitely often from any reachable configuration.

Since $\PP'$ could in principle start from a configuration that
differs from $\init{C}$, there is a risk that $\target{t}$ occurs
infinitely often even though $\init{C} \trans{*} \target{C}$ does
\emph{not} hold in $\PP$. Thus, the role of $\clean{t}$ is to
permanently retire agents until at most $k$ can move to $Q$. This
ensures that $\PP'$ eventually either sets its non retired agents to
$\init{C}$, or gets stuck. The latter only happens if there are
\emph{less} that $k$ agents.

\paragraph*{Proof of the reduction.} Let us prove the claim formally.

\smallskip\noindent$\Rightarrow$) Assume $\init{C} \trans{*}
\target{C}$ holds in $\PP$. Let us show that $\init{C} \not\models
\Diamond \dead{\target{t}}$ in $\PP'$. Let $D$ be some configuration
of $\PP'$ such that $\init{C} \trans{*} D$. We must show that $D$ can
reach $\target{C}$, which allows $\target{t}$ to occur. Since $D$ is
arbitrary, the validity of this claim implies that $\target{t}$ is not
dead at any reachable configuration.

Observe that $\clean{t}$ cannot occur at any reachable configuration
as there are $k$ agents, while $\clean{t}$ requires $k + 1$ agents. By
definition of $S$ and since $D(Q' \setminus \{\out{q}\}) = k$, we have
$D \trans{*} \multiset{k \cdot \free{q}}$. Since $\init{t}$ can occur
from the latter, we have $D \trans{*} \init{C}$. As $\PP'$ contains
all transitions from $\PP$, this implies that $D \trans{*} \target{C}$
in $\PP'$. Hence, $\target{t}$ can occur from there.

\smallskip\noindent$\Leftarrow$) Assume $\init{C} \trans{*}
\target{C}$ does \emph{not} hold in $\PP$. Let us show that $\N^{Q'}
\models \Diamond \dead{\target{t}}$ in $\PP'$. Let $\init{D} \in
\N^{Q'}$. We must argue ``adversarially'' that $\init{D}$ can reach a
configuration $D$ at which $\target{t}$ is dead.

By using ``lossy'' transitions $S$ repeatedly, we can remove all
agents from $Q$. Hence, $\init{D} \trans{*} \multiset{a \cdot
  \free{q}, b \cdot \out{q}}$ for some $a, b \in \N$. If $a > k$, then
using transition $\clean{t}$ repeatedly, we obtain $k$ agents in state
$\free{q}$ and $b + (a - k)$ agents in state $\out{q}$. In other
words, we have
$$\init{D} \trans{*} \multiset{a' \cdot \free{q}, b' \cdot \out{q}}
\text{ where } a' \leq k \text{ and } b' \in \N.$$

If $a' < k$, then \emph{all} transitions are dead and we are
done. Hence, let us assume that $a' = k$. The only enabled transition
at this point is $\init{t}$, which forces $\PP'$ to move to
configuration $D \defeq \init{C} + \multiset{b' \cdot \out{q}}$. Note
that $\clean{t}$ is dead at $D$. Since $\target{C}$ is not reachable
from $\init{C}$ in $\PP$, system $\PP'$ cannot reach $\target{C} +
\multiset{b' \cdot \out{q}}$ either. Moreover, it cannot reach any
configuration larger than $\target{C} + \multiset{b' \cdot \out{q}}$
as the number of agents never changes. Thus, $\target{t}$ is dead at
$D$, which completes the proof.\qed

\end{proof}

\subsection[Missing proofs for Section 5]{Missing proofs for \Cref{{sec:setconfdead}}}\label{app:setconfdead}

\postpresburger*
\begin{proof}
    We use the fact that $C \trans{t} C'$ if{}f $C \in \overline{\sem{\disa{t}}}$ and $C' = C + \effect{t}$, for every $C, C' \in \N^Q$ and $t \in T$.
    Then $C \in \preC_U(\C)$ holds if{}f the following holds:
\begin{align*}
  & \exists C' \in \C :
  \bigvee_{t \in U} \left( C \in \overline{\sem{\disa{t}}} \land C' = C + \effect{t} \right) \\
  \equiv&\
  \bigvee_{t \in U} \left( \left( C + \effect{t} \right) \in \C \land C \in \overline{\sem{\disa{t}}} \right)
  .
\end{align*}
    Similarly, $C' \in \postC_U(\C)$ holds if{}f the following holds:
\begin{align*}
  & \exists C \in \C :
  \bigvee_{t \in U} \left( C \in \overline{\sem{\disa{t}}} \land C' = C + \effect{t} \right) \\
  \equiv&\
  \bigvee_{t \in U} \left( \left( C' - \effect{t} \right) \in \C \land \left( C' - \effect{t} \right) \in \overline{\sem{\disa{t}}} \right)
  .
\end{align*}
\end{proof}

\deathcert*
\begin{proof}
    Let $\C^\omega = \{C^\omega_1, \ldots, C^\omega_k\}$.
    By definition, we have that $\C^\omega$ is a death certificate for $U$ if{}f
    $\Down{\C^\omega} \models \disa{U}$ and $\postC_T(\Down{\C^\omega}) \subseteq \Down{\C^\omega}$.
    We easily have
    \begin{align*}
        \Down{\C^\omega} \models \disa{U}
        \equiv
        \bigwedge_{i=1}^k \bigwedge_{u \in U} \neg \left( C^\omega_i \ge \pre{u} \right).
    \end{align*}
    Using the constraint from the proof of \Cref{prop:post-presburger}, we rewrite the inductivity constraint as follows:
    \begin{align*}
        &\ \postC_T(\Down{\C^\omega}) \subseteq \Down{\C^\omega} \\
        \equiv&\ \forall C' : C' \in \postC_T(\Down{\C^\omega}) \Rightarrow C' \in \Down{\C^\omega} \\
        \equiv&\ \forall C' : \left( \bigvee_{t \in T} \left( \left( C' - \effect{t} \right) \in \Down{\C^\omega} \land \left( C' - \effect{t} \right) \in \overline{\sem{\disa{t}}} \right) \right) \Rightarrow C' \in \Down{\C^\omega} \\
        \equiv&\ \forall C : \bigwedge_{t \in T} \left( \left( C \in \Down{\C^\omega} \land C \in \overline{\sem{\disa{t}}} \right) \Rightarrow \left( C + \effect{t} \right) \in \Down{\C^\omega} \right)
    \end{align*}
    As $\Down{\C^\omega}$ is downward closed, it suffices to check the constraint for all elements in the decomposition of $\Down{\C^\omega}$, i.e., $C^\omega_1$ to $C^\omega_k$.
    This gives us the following formula:
    \begin{align*}
        &\ \bigwedge_{i=1}^k \bigwedge_{t \in T} \left( C^\omega_i \in \overline{\sem{\disa{t}}} \Rightarrow \left( C^\omega_i + \effect{t} \right) \in \Down{\C^\omega} \right) \\
        \equiv&\ \bigwedge_{i=1}^k \bigwedge_{t \in T} \left( C^\omega_i \ge \pre{t} \Rightarrow \bigvee_{j=1}^k \left( C^\omega_i + \effect{t} \right) \le C^\omega_j \right).
    \end{align*}
    Together we obtain the following Presburger formula for $\DeathCert_k(U, \C^\omega)$:
    \begin{align*}
        \left( \bigwedge_{i=1}^k \bigwedge_{u \in U} \neg \left( C^\omega_i \ge \pre{u} \right) \right) \land
       \left( \bigwedge_{i=1}^k \bigwedge_{t \in T} \left( C^\omega_i \ge \pre{t} \Rightarrow \bigvee_{j=1}^k \left( C^\omega_i + \effect{t} \right) \le C^\omega_j \right) \right).
    \end{align*}
    For a fixed $k$, the size of the formula is polynomial w.r.t. the size of the system.
\end{proof}

\subsection[Missing example for Section 6]{Missing example for \Cref{sec:split}}\label{app:split}

\begin{example}
Let $\calP = (Q, T)$ be the replicated system where $Q = \{a_1, \ldots, a_n \} \cup \{b_1, \ldots, b_n\} \cup \{ c \}$ and
$T = U \cup \{ t_c \colon c \tr c \}$ with $U = \{ t_i \colon a_i \, b_i \tr a_{i+1} \, b_{i + 1} \mid 1 \le i < n \} \cup \{ t_n \colon a_n \, b_n \tr a_1 \, b_1 \}$.
Let $\Stage$ be the set of all configurations $C$ where either $C(c) = 0$ or $C(a_i) = C(b_i) = 0$ for all $i$. It is easy to see that no
transition is dead at \emph{every} configuration of $\Stage$, i.e., $\Dead(\Stage) = \emptyset$, but every configuration of $\Stage$ has at least one dead transition:
either $C(c) = 0$ and $t_c$ is dead, or $C(c) > 0$ and all $t_i \in U$ are dead.

Consider the $\omega$-configurations $C^\omega$ and $D^\omega$ defined as follows:
\begin{align*}
    C^\omega  (q) &\defeq \begin{cases} \omega & \text{if $q = c$,    } \\ 0 & \text{otherwise}, \end{cases} &
    D^\omega  (q) &\defeq \begin{cases} \omega & \text{if $q \neq c$, } \\ 0 & \text{otherwise}. \end{cases}
\end{align*}
$C^\omega$ is a death certificate for $U$, and $D^\omega$ is a death  certificate for $\{t_c\}$.
So the pairs $(\multiset{c}, C^\omega)$ and $(\multiset{a_1,\ldots,a_n,b_1,\ldots,b_n}, D^\omega)$ satisfy \ref{itm:split1}--\ref{itm:split3}. It is easy to see that they also satisfy
\ref{itm:split4} and \ref{itm:split5}, and that the only split that can be returned by the procedure is $\{ \Stage \cap \Down C^\omega, \Stage \cap \Down D^\omega \}$. So $\Stage$ is split into only two parts.

We now show that, if condition \ref{itm:split4} or \ref{itm:split5} is dropped, then the splitting procedure might return splits of cardinality $2^n+1$.

Let $\mathcal{M} \defeq \{ C \in \N^Q \mid C(c) = 0 \wedge  \forall \, 1 \le i \le n \colon C(a_i) + C(b_i) = 1 \} \subseteq \Stage$ and, for each $X \in \mathcal{M}$, define the $\omega$-configurations $C^\omega_X, D^\omega_X$ as follows:
\begin{align*}
C^\omega_X(q) &\defeq \begin{cases} \omega & \text{if $q = c$ or $X(q) > 0$,} \\ 0 & \text{otherwise}, \end{cases} &
D^\omega_X(q) &\defeq \begin{cases} \omega & \text{if $X(q) > 0$,           } \\ 0 & \text{otherwise}. \end{cases}
\end{align*}
$C^\omega_X$ is a death certificate for $U$, and $D^\omega_X$ is a death certificate for $\{t_c\} \cup U$. So for every $X \in \mathcal{M}$ the pairs $(X, C^\omega_X)$ and $(X, D^\omega_X)$ satisfy \ref{itm:split1}--\ref{itm:split3}. Since we have $C^\omega \le C^\omega_X$, $D^\omega_X \le C^\omega_X$ and $D^\omega_X \le D^\omega$ for every $X \in \mathcal{M}$, and otherwise the death certificates are pairwise incomparable, condition \ref{itm:split4} is satisfied by all the pairs $(X, D^\omega_X)$, but it is not satisfied by any of the pairs $(X, C^\omega_X)$. It follows that if we drop condition \ref{itm:split4} (removing the reference to \ref{itm:split4} in \ref{itm:split5}), the splitting procedure might find the split $\{ S \cap \Down C^\omega_X \mid X \in \mathcal{M} \} \cup \{\Stage \cap \Down D^\omega\}$.
Without condition \ref{itm:split5}, but with \ref{itm:split4}, it might find the split $\{ S \cap \Down D^\omega_X \mid X \in \mathcal{M} \} \cup \{\Stage \cap \Down D^\omega\}$.
Both splits have $2^n+1$ elements. \defqed
\end{example}

\subsection[Missing proofs for Section 7]{Missing proofs for \Cref{sec:evdead}}

\layerdisabled*
\begin{proof}
We have that $\sem{\disa{U}} = \sem{\dead{U}}$ if{}f $\sem{\disa{U}}$ is inductive, that is
$\postC_T(\sem{\disa{U}}) \subseteq \sem{\disa{U}}$.
We show that
\begin{align*}
\postC_T(\sem{\disa{U}}) \subseteq \sem{\disa{U}}
\equiv
\bigwedge_{t \in T} \bigwedge_{u \in U} \bigvee_{u' \in U} \pre{t} + (\pre{u} \mminus \post{t}) \ge \pre{u'}.
\end{align*}
We start by rewriting the formula as follows:
\begin{align*}
    &\ \postC_T(\sem{\disa{U}}) \subseteq \sem{\disa{U}} \\[3pt]
    \equiv&\ \forall C' : C' \in \postC_T(\sem{\disa{U}}) \Rightarrow C' \in \sem{\disa{U}} \\
    \equiv&\ \forall C' : \left( \bigvee_{t \in T} \left( \left( C' - \effect{t} \right) \in \sem{\disa{U}} \land \left( C' - \effect{t} \right) \in \overline{\sem{\disa{t}}} \right) \right) \Rightarrow C' \in \sem{\disa{U}} \\
    \equiv&\ \forall C : \bigwedge_{t \in T} \left( \left( C \in \sem{\disa{U}} \land C \in \overline{\sem{\disa{t}}} \right) \Rightarrow \left( C + \effect{t} \right) \in \sem{\disa{U}} \right) \\
    \equiv&\ \forall C : \bigwedge_{t \in T} \left( \left( C \in \overline{\sem{\disa{t}}} \land \left(C + \effect{t} \right) \in \overline{\sem{\disa{U}}} \right) \Rightarrow C \in \overline{\sem{\disa{U}}} \right).
\end{align*}
Let $\mathcal{Y}(U,t) \defeq \{ C \mid C \in \overline{\sem{\disa{t}}} \land (C + \effect{t}) \in \overline{\sem{\disa{U}}} \}$. The above formula can be rewritten as:
\begin{align*}
    \forall C : \bigwedge_{t \in T} \left(
    C \in \mathcal{Y}(U,t) \Rightarrow C \in \overline{\sem{\disa{U}}} \right)\
    &\equiv \bigwedge_{t \in T} \mathcal{Y}(U,t)
    \subseteq \overline{\sem{\disa{U}}}.
\end{align*}
Observe that $\mathcal{Y}(U,t)$ is upward closed, as both $\overline{\sem{\disa{t}}}$ and $\overline{\sem{\disa{U}}}$ are upward closed.
Therefore, the inclusion check is between two upward closed sets, which amounts to a comparison of their bases.
We claim that $\Up \mathcal{X}(U,t) = \mathcal{Y}(U,t)$ where $$\mathcal{X}(U,t) \defeq \set{ \pre{t} + (\pre{u} \mminus \post{t}) \mid u \in U }.$$

Let us prove the claim. Let $C = \pre{t} + (\pre{u} \mminus \post{t})
\in \mathcal{X}(U,t)$ for some $u \in U$. Clearly, $C
\in \overline{\sem{\disa{t}}}$ since $C \ge \pre{t}$. Note that $C + \effect{t} =
\post{t} + (\pre{u} \mminus \post{t}) \geq \pre{u}$. Therefore, $C + \effect{t}
\in \overline{\sem{\disa{U}}}$ and consequently $C \in \mathcal{Y}(U,t)$. Since this
also holds for any configuration $C' \ge C$, we obtain $\Up \mathcal{X}(U,t)
\subseteq \mathcal{Y}(U,t)$.

For the other inclusion, let $C \in \mathcal{Y}(U,t)$.
We have $C + \effect{t} \ge \pre{u}$ for some $u \in U$ and hence
follows $\mathcal{Y}(U,t) \subseteq \Up \mathcal{X}(U,t)$ by
$$
C
= C + \effect{t} - \effect{t} \ge \pre{u} + \pre{t} - \post{t}
\ge \pre{t} + (\pre{u} \mminus \post{t})
\in \mathcal{X}(U,t).
$$
We now get the following final formula:
\begin{align*}
      \bigwedge_{t \in T} \bigwedge_{C \in \mathcal{X}(U,t)} \hspace{-8pt}C \in \overline{\sem{\disa{U}}}
      &\equiv \bigwedge_{t \in T} \bigwedge_{u \in U} \bigvee_{u' \in U} \pre{t} + (\pre{u} \mminus \post{t}) \ge \pre{u'}.
\end{align*}
\end{proof}

\end{document}